\documentclass[10pt]{article}
\usepackage{graphicx}
\usepackage{float}
\usepackage{placeins}

\usepackage[letterpaper,top=1in,bottom=1in,left=1in, right=0.5in,footskip=0.5in,marginparwidth=1in]{geometry}
\usepackage{cite}

\usepackage{nameref,hyperref}

\usepackage{microtype}
\DisableLigatures[f]{encoding = *, family = * }

\usepackage{changepage}

\usepackage[aboveskip=1pt,labelfont=bf,labelsep=period,singlelinecheck=off]{caption}

\makeatletter
\renewcommand{\@biblabel}[1]{\quad#1.}
\makeatother

\usepackage{lastpage,fancyhdr,graphicx}
\usepackage{epstopdf}
\fancyheadoffset[L]{2.25in}
\fancyfootoffset[L]{2.25in}

\usepackage{color}
\usepackage[ruled,vlined]{algorithm2e}
\usepackage{bbm}
\usepackage{listings}
\usepackage{amsthm}
\usepackage{placeins}
\usepackage{float}
\usepackage{enumitem}
\usepackage{booktabs}
\usepackage{array}
\usepackage[normalem]{ulem}
\usepackage{tikz}
\usepackage[table]{xcolor}
\usepackage{cancel}   
\usetikzlibrary{tikzmark,calc}
\usepackage{amssymb}
\usepackage{amsmath}
\usepackage{tabularx}

\usepackage{subfig}
\usepackage{changepage}

\usepackage{amsthm}
\newtheorem{definition}{Definition}
\newtheorem{proposition}{Proposition}

\newtheorem{theorem}{Theorem}
\newtheorem{remark}{Remark}
\newtheorem{example}{Example}
\newtheorem{assumption}{Assumption}

\definecolor{Gray}{gray}{.25}

\usepackage{graphicx}

\usepackage{sidecap}

\usepackage{wrapfig}
\usepackage[pscoord]{eso-pic}
\usepackage[fulladjust]{marginnote}
\reversemarginpar

\begin{document}
\vspace*{0.35in}

\begin{center}
{\Large
\textbf\newline{An Entropy-based Framework for Hybrid Coalitions in Game
Theory. Part I: Human Arbitration} }
\newline
\\
Salomé A. Sepúlveda-Fontaine and 
José M. Amigó \textsuperscript{*}
\\
\bigskip
Centro de Investigación Operativa, Universidad Miguel Hernández, 03202 Elche, Spain

\end{center}

\begin{flushleft}
* jm.amigo@umh.es
\end{flushleft}



\section*{Abstract}
Classical Game Theory underpins much of AI and multi-agent research, but hybrid Human–AI systems require a framework in which execution authority can alternate within a digital environment. We introduce Neo-Game Theory, an extension of classical Game Theory for hybrid Human–AI coalitions operating under Virtual Nature, the algorithmic analogue of classical (physical) Nature. The framework combines a lexicographic coalition utility with a delegation rule based on the Jensen--Shannon divergence between Human and AI policies. Two thresholds define agreement, contextual, and disagreement regions. In the contextual region, execution follows a scenario-specific rule. Apart from the theory, in this paper we develop the first regime, Human arbitration, in which the AI learns by observation and frequency matching while the Human retains final execution authority. We establish the axiomatic basis of the framework and characterize a frequency-convergence equilibrium, providing the foundation for later extensions and computational validation.



\section{Introduction}

Game Theory originated with von Neumann and Morgenstern’s
\emph{Theory of Games and Economic Behavior} \cite{Neumann1944TheTO}, which replaced physics-based analogies with a mathematical treatment of strategic interaction. It studies agents who anticipate one another’s responses and reason in terms of payoffs rather than chance.

Delegation has also long been central to economic theory, beginning with principal--agent models in information economics, where a principal transfers decision authority to an agent with distinct objectives or private information \cite{MachStadlerPerez1997}. In industrial organization, strategic delegation emerged prominently in the 1980s, showing how incentive contracts and delegated control can endogenously alter market equilibria in oligopoly settings \cite{FershtmanJudd1987,KopelPezzino2018}.

Classical Game Theory models strategic interaction among agents with fixed decision authority and standard preference structures \cite{Neumann1944TheTO}. Hybrid Human--AI systems differ in a crucial respect: execution authority may alternate endogenously between Human and AI components, so coalition behavior need not be representable by a single smooth preference ordering.
To formalize this setting, we introduce \emph{Neo-Game Theory}. It describes hybrid Human--AI coalitions operating under \emph{Virtual Nature}, the digital counterpart of classical Nature. 

In this paper, we develop the first regime: Human arbitration, in which the AI learns by observation and frequency matching while ultimate execution authority remains with the Human. To this end, we establish the axiomatic basis of the framework and characterize the associated frequency-convergence equilibrium. Related hybrid learning dynamics appear, for example, in the quantum--classical reinforcement-learning setting of \cite{Hamann2022HybridAgent}.

In sum, this paper aims to show why Classical Game Theory is insufficient for hybrid Human--AI coalitions, formulate \emph{Neo-Game Theory} under \emph{Virtual Nature}, and to define and simulate the first regime mentioned: Human arbitration.

\subsection{\label{sec:level2}Related work}

Existing work on hybrid Human--AI systems spans reinforcement learning, robust decision frameworks, and human-in-the-loop architectures \cite{Dellermann2021HybridTaxonomy,Wang2022HybridSurvey}. These approaches motivate the present setting, but they typically either blend Human and AI utilities or treat authority as fixed. In \emph{Neo-Game Theory}, the executing agent is selected at each decision step through a binary delegation rule. The framework therefore studies endogenous control selection under \emph{Virtual Nature}.

This distinguishes the framework from MARL, where utilities are typically aggregated; from Bayesian games, where Nature remains exogenous; and from standard Human--AI collaboration models, which usually combine Human and AI preferences rather than selecting a single executed action at each step. Related work on bounded rationality, knowledge limits, robust decision-making, and hybrid learning motivates the present setting \cite{Dekel1997RationalityAK,EdalatG2018,Rahimian_2022,Aghassi2006Robust,BenTal2009,Hamann2022HybridAgent}, but does not provide the same delegation-based equilibrium structure.

Neo-Game Theory formalizes a class of hybrid Human--AI interactions that existing models do not capture. It studies alternating execution authority, where delegation is determined by policy divergence and behaviour is characterized by empirical regime frequencies.

\subsection{\label{sec:organization} Organization of the paper}

Section~\ref{sec:ClGT} recalls the classical concepts used as the reference point for the hybrid extension. Section~\ref{sec:NeoGameT} introduces the foundations of Neo-Game Theory and \emph{Virtual Nature}. Section~\ref{geneal_formulation} presents the general formulation. Section~\ref{sec:FirstSC} develops Scenario~1, Human arbitration. Section~\ref{sec:Discussion} discusses the implications of the results. Section~\ref{sec:Conclusion} concludes. Furthermore, Appendix~\ref{Glossary} summarizes the notation used in this paper, and Appendix~\ref{pseudocode} reports the base-case pseudocode.


\section{\label{sec:ClGT}Classical Game Theory: a brief review}

The purpose of this section is identifying the assumptions that will be relaxed in Neo-Game Theory.

\subsection{\label{sec:Classical axioms} Classical axiomatic core} The fundamental axioms of classical Game Theory
\cite{Neumann1944TheTO,Bicchieri,towards2018} are the following.

\begin{enumerate}
\item \textbf{Completeness.}
A rational agent can compare any two alternatives $A,B$ and establish a preference or indifference:
$
A\succeq B \;\text{or}\; B\succeq A \;\text{or both }(A\sim B),
$
where $\succeq$ means “at least as good as’’ and $\sim$ denotes indifference.

\item \textbf{Transitivity.}
If $A\succeq B$ and $B\succeq C$, then $A\succeq C$.

\item \textbf{Independence.}
If $A\succeq B$, then for any $p\!\in\!(0,1)$:
$
pA+(1-p)C\succeq pB+(1-p)C.
$

\item \textbf{Continuity (Archimedean axiom).}
If $A\succeq B\succeq C$, there exists $p\!\in\!(0,1)$ such that
$
B\sim pA+(1-p)C,
$
ensuring continuous, rational trade-offs.
\end{enumerate}

These axioms ground the \emph{Expected Utility Theory}~\cite{Neumann1944TheTO},
which guarantees the existence of a utility function $U$ satisfying
\begin{equation}
\label{eq:vNM}  
U(L)=\sum_i p_iU(x_i)
\end{equation}
for any lottery $L=\{(x_i,p_i)\}$.

\subsection{\label{PoliciesDef}Classical Policies}

For reference, below we recall the standard policy definitions:
\vspace{0.2cm}

(i) 
Let $\mathcal{I}_i$ denote player $i$’s information sets and $A(I)$ the actions available at $I\!\in\!\mathcal{I}_i$.  
A \emph{deterministic (pure) policy} is a function
$\pi_i:\mathcal{I}_i\!\to\!\bigcup_{I\in\mathcal{I}_i}\!A(I)$ with
$\pi_i(I)\!\in\!A(I)$.  
Allowing randomization, a \emph{behavioral policy} maps
$\pi_i:\mathcal{I}_i\!\to\!\Delta(A(I))$.  
With perfect recall, behavioral and mixed strategies are outcome-equivalent.

\vspace{0.2cm}

(ii) For a Markov Decision Process \cite{Kuhn1953,OsborneRubinstein1994,Puterman1994} with states $S$, actions $A(s)$, and transition kernel
$P(s'\mid s,a)$, a stationary randomized policy is
$\pi:S\!\to\!\Delta(A(s))$, while a deterministic policy selects a single action $\pi:S\!\to\!A(s)$.  
The induced next-state distribution is $P(s'\mid s,a)$.

\vspace{0.2cm}

\textit{Note: $\Delta(S)$, $\Delta(A(s))$, $\Delta(A(I))$, and $\Delta(A_i)$ denote the corresponding probability simplexes, i.e., the sets of all probability distributions supported on $S$, $A(s)$, $A(I)$, and $A_i$, respectively.}

\subsection{\label{equilibria} Classical Equilibrium}

In the classical setting, an equilibrium is a stable strategic profile under which no player benefits from unilateral deviation. Nash, Stackelberg, and mixed-strategy equilibria provide the standard reference concepts \cite{Neumann1944TheTO,Nash1950,Nash1951,vonStackelberg1934,BasarOlsder1999,Walker2016}. These concepts remain the comparison point throughout the paper, but they will not directly apply once execution authority alternates and the coalition no longer behaves as a single standard rational agent.
\subsection{\label{subsec:Utility function} Classical utility function}

A utility function represents an agent’s preference ordering over consequences \cite{Bicchieri}. Under the von Neumann--Morgenstern axioms (see Section~\ref{sec:Classical axioms}), preferences over lotteries admit the expected-utility representation in Equation~\eqref{eq:vNM} \cite{Neumann1944TheTO} where \(u_i\) is the vNM utility unique up to a positive affine transformation (\(u_i\mapsto a\,u_i+b,\;a>0\)).

This classical representation is retained for the individual Human and AI utilities \(U_H\) and \(U_{AI}\); what changes later is the coalition-level structure induced by alternating execution authority.  For abstract or symbolic payoffs, utilities are often normalized (e.g., \(u_i\in[0,1]\)) since equilibrium reasoning depends only on ordinal, not cardinal, magnitudes.

\subsection{\label{sec:ClassValuefunction} Classical value function}

In dynamic settings, the value function aggregates current rewards into an expected discounted return \cite{Puterman1994,Bellman1957}. We recall the classical state-value form in Equation~\eqref{eq8} only because the hybrid model will later modify the reward-generation mechanism, not the recursive logic of dynamic evaluation itself.

Formally, in a sequential or dynamic setting, the \emph{state–value function} for player~\(i\) is defined as
\begin{equation}
v_i(s)
=\mathbb{E}\!\left[\sum_{t=0}^{\infty}\gamma^t\,r(s_t,a_{i,t})
\;\middle|\;s_0=s,\,\pi_i\right],
\label{eq8}
\end{equation}
where \(r(s_t,a_{i,t})\) denotes the immediate reward at time~\(T\), \(\gamma\in(0,1)\) is the discount factor ensuring convergence, and \(\pi_i\) is the policy mapping each state to an action distribution~\cite{Puterman1994}.  
This recursive formulation embodies Bellman’s principle of optimality~\cite{Bellman1957}, which asserts that an optimal strategy must remain optimal at every subsequent state reachable under that strategy.

\section{\label{sec:NeoGameT}Neo-Game Theory: foundations}


\emph{Neo-Game Theory} extends classical Game Theory to hybrid Human--AI systems interacting with  \textbf{\textit{Virtual Nature}}. In particular, the hybrid Human--AI setting departs from the classical baseline because execution authority may alternate and coalition-level behaviour need not satisfy the smooth trade-offs presupposed by expected-utility analysis \cite{Neumann1944TheTO,Simon1955,Machina1992,Amodei2016}.


\subsection{Departure from Classical equilibrium and continuity}

\label{non_archi}

\textbf{(a) Departure from classical equilibrium.} The following normal-form example shows why the realized outcome may differ from the classical equilibrium. In the classical simultaneous-move or Stackelberg interpretation, the outcome is determined by best-response logic in a fixed game. Here, instead, execution is selected endogenously by the delegation rule through policy divergence.

\begin{example} To ground the hybrid interaction in a concrete strategic setting, consider a simple two-action game where both the Human and the AI choose between
actions $A$ and $B$. Utilities differ across agents:
\[
\begin{array}{c|cc}
 & A_{AI} & B_{AI} \\
\hline
A_H & (3,1) & (0,2) \\
B_H & (2,0) & (1,3)
\end{array}
\]
Each cell lists $(U_H,U_{AI})$. The Human prefers $(A_H,A_{AI})$, while the AI prefers $(B_H,B_{AI})$, creating strategic misalignment.
\end{example}

In the simultaneous-move formulation, the unique Nash equilibrium is $(B_H,B_{AI})$, since $B_{AI}$ is a dominant strategy for the AI and $B_H$ is the Human’s best response to it. In a Stackelberg formulation with the Human as leader, anticipating the AI’s best response also yields the outcome $(B_H,B_{AI})$. In both cases, execution follows from strategic best-response reasoning within a fixed game.

Thus, the realized outcome need not coincide with the Nash or Stackelberg equilibrium $(B_H,B_{AI})$: execution authority depends on policy divergence and threshold-based delegation rather than strategic dominance or commitment. This illustrates how hybrid delegation may generate outcomes outside classical equilibrium concepts.

\vspace{0.5cm}

\textbf{(b) Departure from Archimedean Principle.} Classical rationality presumes that all players satisfy the four axioms of Section \ref{sec:Classical axioms} and that this is common knowledge. In hybrid Human--AI settings, however, the fourth axiom (usually called the Archimedean axiom) may fail at the coalition level because human adjustment can be gradual while AI behavior may change through discrete optimization and hard constraints, thereby producing non-Archimedean effects \cite{Fishburn1982,Hammond1976,Amodei2016}. This motivates the revised axiomatic structure introduced later in this section.

\begin{example}[Failure of transitivity under alternating control]
\label{ex:transitivity-failure}
Let $U_H$ and $U_{AI}$ denote Human and AI utilities on outcomes $A,B,C$:
\[
U_H(A)=0.9,\;U_H(B)=0.7,\;U_H(C)=0.6; \;\;
U_{AI}(A)=0.2,\;U_{AI}(B)=0.5,\;U_{AI}(C)=0.9.
\]
When the Human executes, preferences are $A\succ_H B\succ_H C$.
When the AI executes, preferences are $C\succ_{AI} B\succ_{AI} A$. If execution alternates switches over $T$, the coalition exhibits the following: 
\[
A\succ B,\quad B\succ C,\quad C\succ A.
\]

This formal non-transitivity result is stated and proved in Proposition~\ref{fail_transivity}
\end{example}

\begin{definition}[Union-over-time revealed preference]
Let $\{\succ_T\}_{T\ge0}$ denote the preference relations induced by
the executing agent at time $T$. The aggregated revealed relation is

\begin{equation}
A \mathcal{R} B
\;\Longleftrightarrow\;
\exists T \ge 0 \text{ such that } A \succ_T B .
\end{equation}

\end{definition}
\begin{proposition}[Alternating control induces failure of transitivity]
\label{fail_transivity}
Let $U_H$ and $U_{AI}$ induce individually transitive preference orderings $\succ_H$ and $\succ_{AI}$ for the Human and the AI, respectively. Suppose execution authority evolves over time according to the contextual delegation rule described in Section~\ref{entropythreshold}
and formalized in Equation~(\ref{eq:lambda_rule}), with $\lambda_T \in \{0,1\}$ selecting the active executor at instant $T$.
Then the coalition’s time-aggregated revealed comparison relation need not satisfy transitivity. In particular, alternating control can generate a cycle 
\[
A \succ B,\qquad B \succ C,\qquad C \succ A,
\]
even though each agent’s own preference ordering is internally transitive.

\end{proposition}

\begin{proof}
For each time $T$, let $\succ_T$ denote the revealed comparison relation induced by the executor selected via the contextual delegation rule described in Section~\ref{entropythreshold} and Equation~(\ref{eq:lambda_rule}) below. Thus, if

\[\succ_T =
\begin{cases}
\succ_H, & \text{if } \lambda_T(s) = 1,\\
\succ_{AI}, & \text{if } \lambda_T(s) = 0,
\end{cases} \]
then the time-aggregated relation is
\[\mathcal{R} := \bigcup_{T\ge 0} \succ_T.\]

Assume there exist three alternatives $A,B,C$ such that
\[U_H(A) > U_H(B) > U_H(C), \quad U_{AI}(C) > U_{AI}(A).\]
If at time $T$ we have $\lambda_T(s)=1$, then $\succ_T=\succ_H$
and hence $A \mathcal{R} B$.
If at time $T+1$ we again have $\lambda_{T+1}(s)=1$,
then $B \mathcal{R} C$.
If at time $T+2$ we have $\lambda_{T+2}(s)=0$,
then $\succ_{T+2}=\succ_{AI}$ and hence
$C \mathcal{R} A$.

Therefore $\mathcal{R}$ contains the cycle
$A \mathcal{R} B, B \mathcal{R} C,  C \mathcal{R} A$
and thus may fail to be represented by any single transitive ordering.
The failure of transitivity arises from time-dependent
execution authority determined by Equation~(\ref{eq:lambda_rule}),
not from inconsistency within $U_H$ or $U_{AI}$.
\end{proof}

This phenomenon resembles rank-dependent or imprecise-preference effects.

In this work, the \textbf{non-Archimedean feature} arises from the switching (guided by the regime priority) that preceds the comparison of utilities \cite{Vilaseca2008}, rather than from lottery weighting or belief imprecision. Hence, the coalition behavior is best represented by a lexicographic utility \cite{Fishburn1982,Hammond1976} (the execution regime is selected before utility is evaluated)

Both utilities $U_H$ and $U_{AI}$ (Human and AI utilities respectively, e.g., comfort, goals, safety, or efficiency) \textbf{individually, remain standard real-valued} satisfying the Archimedean axiom and regularity assumptions stated earlier. Formally,
\begin{equation}
\label{eq:lex_order}
(\rho_1,ut_1) \succ_{\mathrm{lex}} (\rho_2,ut_2)
\Longleftrightarrow
(\rho_1>\rho_2)
\text{ or }
(\rho_1=\rho_2 \text{ and } ut_1>ut_2).
\end{equation}

Within a fixed execution regime, the learning and dynamic programming recursion operate on the scalar realized utility. Across regimes, coalition outcomes are compared lexicographically, with regime priority preceding the comparison of utilities \cite{Vilaseca2008}.

Our formulation differs from nonlinear \emph{utility-blending} models in multi-agent reinforcement learning, welfare economics, and Human--AI collaboration~\cite{utilitysurvey2023,responsibility2022,synergy2024}. Those approaches combine Human and AI payoffs through concave or convex weights to represent \emph{a posteriori} synergy or shared responsibility. Here, by contrast, only the utility associated with the executed action is realized at each decision step, and the nonlinearity arises from threshold-based regime switching rather than payoff aggregation. Therefore, although each agent individually satisfies the classical axioms of Expected Utility Theory (Section~\ref{sec:ClGT}), these implications need not hold at the coalition level under alternating decision authority.

Figure~\ref{Figure1} is included only as a visual illustration of the lexicographic order and the non-Archimedean effect induced by switching. 
As the number of iterations (i.e., the number of move counts during the entire game) increases, these coalition-level discontinuities become less visible. Policies show poor convergence due to the small number of iterations.

\begin{figure}[H]
    \centering
    \includegraphics[width=1\textwidth]{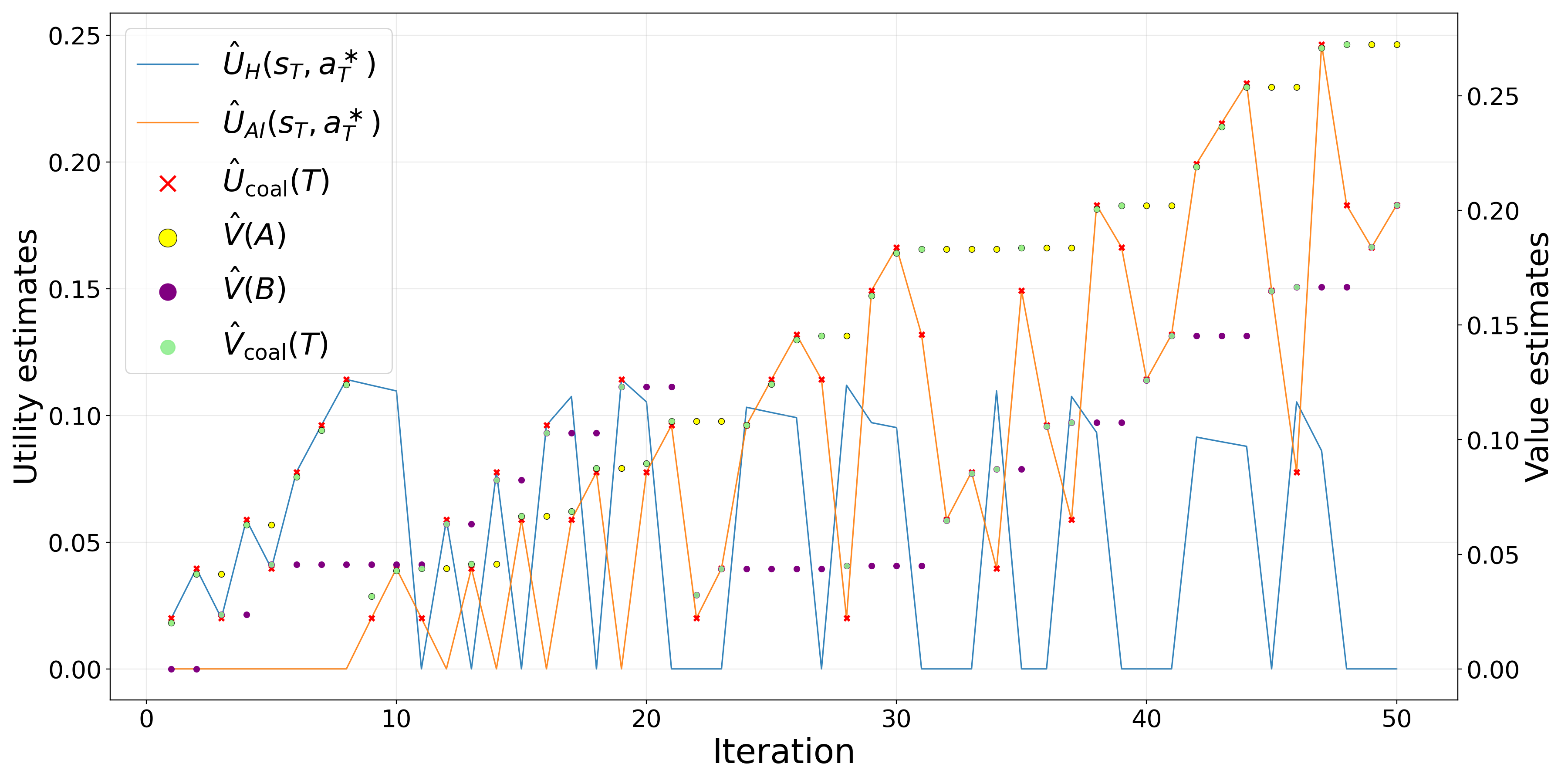}
    \caption{Illustrative example of utilities and value functions under lexicographic regime switching at 50 iterations. The figure is included only to visualize the discontinuous coalition structure induced by $\lambda_T$.}
    \label{Figure1}
\end{figure}
\FloatBarrier

\subsection{\textit{Neo-}Axiomatic principles}

\emph{Neo-Game Theory} extends classical Game Theory to hybrid Human--AI systems where, taken together Axioms 1-4 from Section~\ref{sec:Classical axioms} and Section~\ref{non_archi}, imply:
\begin{enumerate} [label=(\roman*)] 
    \item The Archimedean axiom fails because hybrid preferences lack smooth trade-offs;
    \item Decision authority alternates probabilistically between Human and AI;
    \item Utility and value functions may be discontinuous;
    \item The coalition need not admit a single stable rational ordering;
    \item Completeness can fail because the coalition does not act under a single coherent preference relation.
\end{enumerate}


\subsection{Virtual Nature $\mathcal{V}$}
\label{virtualnature}

\begin{definition}[Virtual Nature]
In \textit{Neo}-Game Theory, \emph{Virtual Nature} $\mathcal{V}$, as the transition function
\[
\mathcal{V}:S\times A\to\Delta(S),
\]
where $\Delta(S)$ denotes the set of probability distributions over states. 
It captures algorithmic stochasticity, hard thresholds, and discontinuities generated by digital rules, and therefore plays the role that classical Nature plays in standard games.
\end{definition}

$\mathcal{V}$ is interpreted then, as unconscious, having an optimization-driven behaviour and although random, yet influenced by embedded human biases and algorithmic priors.

\subsection{Hybrid games}
\label{hybridgame}

\begin{definition}[Hybrid game]
\label{def:hybridgame}
A hybrid game is a game where one of the players is Human and the other is an \textit{entity} in $\mathcal{V}$ or a hybrid coalition, so the game can be played by two players or more, interaction with the Physical Nature or Virtual Nature. Formally, a hybrid game induced by $\mathcal{V}$ is the tuple

\[ \mathcal{N}(\mathcal{V}) = (S, A, U, \lambda, P_{\mathcal{V}}), \]
where
\begin{itemize}
\setlength{\itemsep}{4pt}

\item $S$ is the state space; the system evolves over discrete time steps $T=0,1,2,\ldots$, and $s_T\in S$ denotes the state at instant $T$ ($t$ denoted index over trajectories);

\item $A$ is the common executed action space. The Human and the AI may have (possibly overlapping) proposal sets $A_H \subseteq A$ and $A_{AI} \subseteq A$, and we assume $A = A_H \cup A_{AI}$. At each time step, the feasible executed action set 
$\mathcal{A}_{\mathrm{exec}}(s)=\{a_{H,T}(s),a_{AI,T}(s)\}\subseteq A$ contains the proposals of the Human ($a_{H,T}(s)$) and AI ($a_{AI,T}(s)$) and $a_T^*$ represent the executed action.

\item  $U=\{U_H,U_{AI}\}$, where $U_i:S\times A\to\mathbb{R}$ denotes the individual utility of agent $i\in\{H,AI\}$;

\item $\lambda_T\!\in\!\{0,1\}$ the delegation factor indicating which action is executed.

\item $P_{\mathcal{V}} : S \times A \to \Delta(S)$ is the transition kernel generated by Virtual Nature $\mathcal{V}$, i.e.

$
P_{\mathcal{V}}(\cdot \mid s,a) =
P_{\mathcal{V}}\!\bigl(s_{t+1} \in \cdot \mid s_t = s,\; a_t = a\bigr)
$
\end{itemize}
\end{definition}

Hybrid games can be characterized as follows. 

\begin{itemize}
    \item Not all players are rational since AI agents lack genuine agency~\cite{huang2024positionpap,asocialpath}.

    \item  Human recognizes this asymmetry but cannot evaluate coalition-level rationality.

    \item  In our mathematical construct, the game action selection follows a probabilistic delegation rule rather than a deterministic optimization.

    \item  If one of the players does not proposed any action, the action executed is the one proposed by the other player. Otherwise, the game enters in the contextual zone defined later in Section~\ref{entropythreshold}.

    \end{itemize}


\label{game_ag_V} 

\begin{definition}[Games against Virtual Nature]
\label{def:game_ag_V} 
Games against Virtual Nature are the basic hybrid games (the analogous of games against Nature defined in classical Game theory). One player can be a Human agent, an entity in $\mathcal{V}$ or a hybrid coalition. The opposing player \emph{must be} an \textit{entity in} $\mathcal{V}$ (another AI agent, a bot, etc.)
\end{definition}

Thus, a game against $\mathcal{V}$ is a type of hybrid game defined as a stochastic game equipped with an endogenous delegation mechanism that determines which internal proposed action (i.e. whose player's action) is executed at each instant.


\subsection{Hybrid coalition}
\label{def:hybridcoalition}

\begin{definition}[Hybrid coalition]
\label{def:hybridcoalition}

A  hybrid coalition is composite by a Human and an \textit{entity} in $\mathcal{V}$ in which:

\begin{itemize}
    \item Each member of the coalition plays offering its own preference but at the end, only one is executed.

    \item AI initially depends on Human oversight but gradually refines its own strategies, necessitating continual supervision
~\cite{huang2024positionpap,Mart_nez_Fern_ndez_2021,intheloop}.
    \item The concept of negotiation can not be modelized as in Classical Game Theory.

    \item AI may require additional stabilization when operating under abrupt regime switching or discontinuous policy updates.

    \item  As rationality is no longer standard, optimization becomes layered, and convergence is not guaranteed under regime switching.
\end{itemize}

\end{definition}


\section{Formulation of \textit{Neo}-Game Theory}
\label{geneal_formulation}

In this section we introduce the general formulation common to all regimes of Neo-Game Theory. The Human arbitration scenario developed in the next section is obtained by specifying this generic structure in a particular way.

\begin{assumption}[Regularity conditions]
\label{ass:lambda}
\leavevmode\par
\begin{enumerate}[label=(\roman*),leftmargin=*,itemsep=0.2em,topsep=0.3em]

\item The state space $S$ and the action sets $A_i$ may be either
(a) compact and convex subsets of $\mathbb{R}^n$ (continuous case), or
(b) finite sets (discrete case).
Both formulations are admissible within this framework.
\emph{For the simulations, a finite-state, finite-action instance is used for clarity:
specifically, a two-state system $S=\{A,B\}$ with two discrete actions per agent.}

\item Individual utilities $U_H(s,a),U_{AI}(s,a) $ are bounded and Lipschitz continuous in $a_{H,T}, a_{AI,T}$ respectively.

\item Each policy update is a Lipschitz best response to the last executed action.

\end{enumerate}
\end{assumption}

\begin{definition}[Owner equals origin]
For each time $T$ define
\[
o_T =
\begin{cases}
\mathrm{AI}, & a_T^* = a_{AI,T},\\
\mathrm{H}, & a_T^* = a_{H,T}.
\end{cases}
\]
Owner is a credit label distinct from execution authority $\lambda_T(s)$.
\end{definition}


\subsection{\label{sec:Entropy} Entropy: Jensen--Shannon divergence.}

Entropy quantifies uncertainty in a random variable and serves as a positive functional over probability distributions.
We adopt Shannon entropy $S(Z)$ for a discrete variable $Z$ with mass function
$p$:
\[
S(Z)=-\sum_z p(z)\log p(z),
\]
where $\log$ is taken in any base $b>1$.
For an overview of entropy applications in AI and machine learning,
see~\cite{SepulvEntropyRev}.

Let $p_H(\cdot\mid s_T,T)$ and $p_{AI}(\cdot\mid s_T,T)$ denote the empirical
Human and AI policies at time $T$ in state $s_T$. Their corresponding entropies are
\begin{align}
S\!\big(p_H(\cdot\mid s_T,T)\big)
&= -\sum_{a\in A} p_H(a\mid s_T,T)\log p_H(a\mid s_T,T), \\
S\!\big(p_{AI}(\cdot\mid s_T,T)\big)
&= -\sum_{a\in A} p_{AI}(a\mid s_T,T)\log p_{AI}(a\mid s_T,T).
\end{align}

High Human entropy and low AI entropy
$\big(S(p_H(\cdot\mid s_T,T)) \gg S(p_{AI}(\cdot\mid s_T,T))\big)$ indicate great Human uncertainty and great AI confidence; convergence of these entropies indicates increasing agreement. The empirical policies are defined from observed frequencies:

\begin{equation}
\label{empirical_policies_H}%
\displaystyle
p_H(a\mid s_T,T)
:= \frac{N_H(s_T,a,T-1)}{N_H(s_T,T-1)},
\end{equation}

\begin{equation}
\label{empirical_policies_AI}%
\displaystyle
p_{AI}(a\mid s_T,T)
:= \frac{N_{AI}(s_T,a,T-1)}{N_{AI}(s_T,T-1)},
\qquad 
\end{equation}
where
\[
N_i(s_T,T-1) := \sum_{a\in A} N_i(s_T,a,T-1)
\]
counts how many times agent $i\in\{H,AI\}$ selected action $a$
in state $s_T$ up to time $T-1$. To measure distributional similarity, we use the symmetric and bounded Jensen-Shannon ($D_{JS}$) divergence~\cite{LinDJS,Nielsen2020}:
\begin{equation}
D_{\mathrm{JS}}(P\parallel Q)
=S\!\left(\tfrac{P+Q}{2}\right)
-\tfrac{1}{2}S(P)
-\tfrac{1}{2}S(Q).
\label{JSdivergence}
\end{equation}

The Jensen–Shannon divergence is bounded above by $\log 2$.
Using base-2 logarithms normalizes the divergence to the interval $[0,1]$
\cite{LinDJS,Endres2003}. Therefore, according to Equation~(\ref{JSdivergence}),
\begin{eqnarray}
\label{eq16}
D_{JS}\!\big(p_H(\cdot\mid s_T,T)\,\Vert\,p_{AI}(\cdot\mid s_T,T)\big)
&=&
S\!\left(\frac{p_H(\cdot\mid s_T,T)+p_{AI}(\cdot\mid s_T,T)}{2}\right)
\nonumber\\
&&-\frac{1}{2}S\!\big(p_H(\cdot\mid s_T,T)\big)
-\frac{1}{2}S\!\big(p_{AI}(\cdot\mid s_T,T)\big),
\label{eq:DJS}
\end{eqnarray}
which will be shortened to
\begin{equation}
\label{D_JS_T}
D_{JS}^T := D_{JS}\big(p_H(a \mid s,T),\, p_{AI}(a \mid s,T)\big),
\end{equation}
when the time $T$ is important.

In this paper, we use the Jensen--Shannon divergence as an operational alignment statistic between the empirical Human and AI policies. This choice is motivated by three properties that are directly relevant for delegation: symmetry, boundedness, and, with base-$2$ logarithms, normalization on $[0,1]$ \cite{LinDJS, Nielsen2020, Endres2003}. These properties make the thresholds 
directly interpretable across states and time and allow the same scale of policy discrepancy to be used throughout the delegation rule. Moreover, precisely because $D_{\mathrm{JS}}\in[0,1]$, it can also be used directly in Equation~\eqref{lambda_context} as a probability for the ex-ante modelling of Human arbitration in the contextual region, without any additional rescaling or transformation of the original divergence between Human and AI policies. In this sense, the same quantity serves both as a regime-classification statistic and as a probabilistic driver of contextual Human intervention. In the present framework, $D_{\mathrm{JS}}$ is therefore introduced as an observable and scale-stable proxy for policy misalignment, not as a welfare functional and not as a Lyapunov function for the learning dynamics. Stability in the model instead derives from the stochastic-approximation recursion, compactness of the policy space, bounded martingale noise, and the associated ODE/differential-inclusion arguments \cite{RobbinsMonro1951,KushnerYin2003,Benaim1999,BenaimHofbauerSorin2005}.

\subsection{Scenario-dependent delegation}
\label{entropythreshold}

At each time $T$ and state $s_T$, the Human and the AI propose actions according to $a_{H,T}\sim p_H(\cdot\mid s_T,T)$ and $a_{AI,T}\sim p_{AI}(\cdot\mid s_T,T)$. Delegation is then determined by a two-threshold rule based on $D_{JS}^T$, encoded by the binary variable $\lambda_T(s)$.

\begin{equation}
\label{eq:lambda_rule}
\lambda_T(s) =
\begin{cases}
\lambda_{\mathrm{agree}} & \text{if}\; D_{JS}^T \le \alpha_{\mathrm{agree}},\\[4pt]
\lambda_{\mathrm{ctx}}(\cdot) & \text{if}\; \alpha_{\mathrm{agree}} < D_{JS}^T < \alpha_{\mathrm{disagree}},\\[4pt]
\lambda_{\mathrm{disagree}} & \text{if}\; D_{JS}^T \ge \alpha_{\mathrm{disagree}},
\end{cases}
\end{equation}
where $\lambda_{\mathrm{agree}}$, $\lambda_{\mathrm{disagree}}$, $\lambda_{\mathrm{ctx}(\cdot)}\in\{0,1\}$. $\lambda_{\mathrm{ctx}}$ denotes the contextual delegation rule defined in Section~\ref{delegation_sc1} and $\alpha_{agree}$, $\alpha_{disagree}$ are the ex ante fixed probabilities that generate the range in which $D_{JS}$ is used to select one of the three regimes. We always assume $\alpha_{agree} <\alpha_{disagree}$.

In the Human arbitration scenario studied in this work, Definition~\ref{twothreshold} gives the scenario-specific form of $\lambda_T(s)$.


\subsection{Neo Utility function}
\label{neoutility}

We now define the induced coalition utility $U_{\mathrm{joint}}$ from the individual utilities $U_H$ and $U_{AI}$. Formally, \(U:C\!\to\!\mathbb{R}\) satisfies
\(x\succeq y \iff U(x)\!\ge\!U(y)\) for all \(x,y\in C\).

\begin{equation}
\label{eq:Ujoint}
U_{\mathrm{joint}}(\lambda_T(s),s_T,a_T^*)
=
\begin{cases}
U_{AI}(s_T,a_T^*) 
& \text{if $a_T^*=a_{AI,T}$ } , \\[4pt]
U_H(s_T,a_T^*)
& \text{if $a_T^*=a_{H,T}$ }, \\[6pt]
U_{\lambda_{\mathrm{ctx}}(\cdot)}(s_T,a_T^*)\
& \text{if } \alpha_{\mathrm{agree}} < D_{JS}^T < \alpha_{\mathrm{disagree}}.
\end{cases}
\end{equation}

\label{subsec:utility_estimates}

We distinguish the \emph{structural} utilities $U_H,U_{AI}:S\times A\to\mathbb{R}$ (from Definition~\ref{def:hybridgame}) from the \emph{learned} utility traces $\widehat U_H,\widehat U_{AI}$ updated online.

Utility credit follows the executed action. The agent whose action is executed receives reinforcement through an exponentially weighted trace (EWMA):
\begin{equation}
 \label{estimated_utility}   
\widehat U_i^{\,T+1}(s,a)
=
(1-\beta)\,\widehat U_i^{\,T}(s,a)
+
\beta\,\mathbf{1}\{i=o_T \land (s,a)=(s_T,a_T^\ast)\},
\qquad i\in\{H,AI\}.
\end{equation}
\noindent
The structural utilities $U_H,U_{AI}$ define preferences; $\widehat U_{i,T}$ are online estimates used by the learning rule and value recursion.

\subsection{Neo-Policies: distributions and asymptotic notation}

Next, we specify how policies updates are represented in the hybrid framework. Proposed actions determine delegation through $D_{JS}$, whereas the executed action $a_T^\ast$ determines reward and state transition.

When Virtual Nature $\mathcal{V}$ is fixed, we write $P_{\mathcal{V}} \equiv P$ for brevity. We use a state-indexed notation. For each state $s_j\in S$: $a_{H,T}(s_j)\in A$ and $a_{AI,T}(s_j)\in A$ denote the Human’s and AI’s proposed actions at $s_j$; $p_{H,s_j}$ and $p_{AI,s_j}$ denote their empirical action distributions over $A$; $p_i(s_j,T)$ is the instantaneous empirical policy at time $T$.

Because delegation is triggered by $D_{JS}$ between Human and AI policies, the long-run behavior of the system is naturally characterized by empirical frequencies rather than by best responses.


\subsubsection{Empirical frequencies}

\textbf{\textit{a) Agreement frequency.}} Define the empirical agreement-by-closeness frequency after $T$ realized decision turns as
\begin{equation}
\label{eq:f_agree}
f_{\mathrm{agree}}(T)
=
\frac{1}{T}
\sum_{t=1}^{T}
\mathbf{1}\!\left\{
D_{JS}^{t} \le \alpha_{\mathrm{agree}}
\right\}.
\end{equation}
It measures the empirical fraction of decision turns up to time $T$ for which the divergence between the Human and AI policies lies within the agreement band $D_{JS}^{t} \le \alpha_{\mathrm{agree}}$.

\textbf{\textit{b) Contextual frequency.}} Define the empirical contextual-regime frequency as
\begin{equation}
f_{\mathrm{ctx}}(T)
=
\frac{1}{T}
\sum_{t=1}^{T}
\mathbf{1}\!\left\{
\alpha_{\mathrm{agree}} < D_{JS}^{t} < \alpha_{\mathrm{disagree}}
\right\}.
\end{equation}
It measures the empirical fraction of decision turns for which divergence lies strictly between the two thresholds, corresponding to the contextual delegation regime.

\textbf{\textit{c) Disagreement frequency.}}
Define the empirical disagreement frequency as
\begin{equation}
f_{\mathrm{disagree}}(T)
=
\frac{1}{T}
\sum_{t=1}^{T}
\mathbf{1}\!\left\{
D_{JS}^{t} \ge \alpha_{\mathrm{disagree}}
\right\}.
\end{equation}
It measures the empirical fraction of decision turns for which divergence exceeds the disagreement threshold, corresponding to persistent misalignment.

Since  $0 \le \alpha_{\mathrm{agree}} < \alpha_{\mathrm{disagree}} \le 1$, the $D_{JS}$ range $[0,1]$ is partitioned into three mutually exclusive regions, so that, for every $T$,
\begin{equation}
\label{freq_summa}    
f_{\mathrm{agree}}(T)+f_{\mathrm{ctx}}(T)+f_{\mathrm{disagree}}(T)=1.
\end{equation}


\subsubsection{Policy learning}

\label{policy_update}

\noindent For each state $s$ and time $T$, $p_i(s,T)\in\Delta(A_i)$ denotes agent $i$'s instantaneous empirical policy, i.e., a probability vector that represents the instantaneous empirical action distribution of agent $i$ in state $s$. Decisions and divergence computations at time $T$ are based on $p_i(s,T)$, i.e., the distribution used at decision time $T$ to compute delegation and divergence. When the sequence converges, we write
\begin{equation}
   \label{limit_policies}
\pi_i(s)=\lim_{T\to\infty}p_i(s,T),
\end{equation}
where $\pi_i(s)$ denotes the corresponding limiting stationary policy. Thus $p_i(s,T)$ governs finite-time execution, whereas $\pi_i(s)$ is used only for the asymptotic description of long-run behavior. We assume that $p_{i}(s,T)$ converges when $T\to\infty$, so that $\pi_{i}(s)$ is a well-defined probability. If the sequence of probability distributions $p_i(s,T)$ converges, then the limit itself is a valid probability distribution:
\[
\sum_a p_i(s,T)(a)=1\ \; \forall T
\quad\Rightarrow\quad
\sum_a \pi_i(s)(a)=1.
\]
Thus the limit remains a valid stationary probability measure.


\begin{remark}
\textbf{(i) Asymptotic regime}. 
All equilibrium statements are formulated in the large-sample limit. For each state $s$ visited infinitely often, the empirical policies $p_i(\cdot\mid s,T)$, $i\in\{H,AI\}$, follow a stochastic approximation recursion whose limiting behavior is governed by the differential inclusion (see Remark~\ref{justification_v_sthicastic} below).
\textbf{(ii) Finite-$T$ convention}. When a state $s$ is first observed, $P(\cdot\mid s)$ is initialized as uniform (Dirichlet$(1)$ smoothing); this initialization does not affect equilibrium results, which hold asymptotically. The convention is stated for the finite-state, finite-action setting; for continuous spaces, the same reasoning applies after discretization or through empirical measures on measurable partitions.
\end{remark}

Utilities record influences history, whereas policies evolve directly from executed actions. After each interaction, the empirical policy $p_i(s,T)$ is updated toward the executed action $a_T^\ast$ according to Equation~\eqref{monroeq}. We consider both constant-gain and Robbins--Monro step-size schedules, since they generate different transient behavior while the same conceptual learning mechanism.

\begin{align}
\label{monroeq}
p_i^{(T+1)}(s)
   =(1-\eta_T)\,p_i^{(T)}(s)+\eta_T\,\delta_{a_T^\ast},
   \qquad i\in\{H,AI\},\ \eta_T>0,
\end{align}
where $\delta_{a_T^\ast}$ in the policy update is the Dirac/one-hot action vector for the executed action $a_T^\ast$. If we have (as in the simulations) two actions ($a_0$, $a_1$), then 
\[
\delta_{a_T^\ast} =
\begin{cases}
(1,0), & \text{if } a_T^\ast = a_0, \\
(0,1), & \text{if } a_T^\ast = a_1.
\end{cases}
\]

The qualitative behavior of this recursion depends on the step-size schedule. Because $A_H,A_{AI}\subseteq A$, both policies are represented on the common executed-action space $\Delta(A)$, while the proposal sets restrict which actions each agent may sample. Human policy updates occur only when $D^T_{JS} \leq \alpha_{\mathrm{agree}}$, thereby giving the AI a learning advantage in the hybrid update scheme.

Consider a \textit{constant step size. } Under a constant step size $\eta_T \equiv \eta \in (0,1)$, the executed action $a_T^\ast$ and the delegation rule $\lambda_T(s)$ are measurable functions of the current policy state $x_T=(p_H(T),p_{AI}(T))$ and the current state $s_T$. Hence $x_{T+1}$ depends only on $(x_T,s_T)$, so the policy process $\{x_T\}$ is a time-homogeneous Markov chain on the compact policy space $\mathcal X$. If this chain is irreducible and aperiodic, then it admits a unique invariant distribution, and convergence holds in the ergodic sense through occupation measures. Thus, in the constant-gain regime, almost-sure convergence need not occur; instead, the process stabilizes in distribution around its stationary/invariant law, and equilibrium statements are interpreted through long-run empirical frequencies.

Under diminishing step sizes, $p_i(\cdot\mid s,T)$ approaches the internally chain transitive set of the limiting dynamics and converges almost surely to a limit distribution $\pi_i(\cdot\mid s)$ when that set is a singleton.

\begin{assumption}
    \label{eta_assumption}
    For diminishing step sizes, we assume that the learning rate $\eta$ satisfies (i) $\sum_T \eta_T = \infty$, and (ii) $\sum_T \eta_T^2 < \infty$. 
\end{assumption}

If the step sizes satisfy the previous assumption \cite{RobbinsMonro1951}, the recursion becomes a Robbins--Monro stochastic-approximation scheme.
Under standard ODE-method results \cite{RobbinsMonro1951,KushnerYin2003,BenaimHofbauerSorin2005},
the iterates converge almost surely to an invariant set
of the limiting ODE.

Policy adaptation is asymmetric: the AI updates toward the executed action in all regimes, whereas the Human updates only in the agreement region $(D_{JS}^T\le\alpha_{\mathrm{agree}})$. Thus the AI gradually aligns with realized behavior, while Human reinforcement occurs only under low-divergence alignment.

\subsection{Bellman recursion: reward and value function}
\label{neovalue}

The value function links static utility with dynamic optimization through Bellman’s equation. In this framework, the proposed actions determine delegation through $D_{JS}$, whereas the executed action $a_T^\ast$ determines the reward and state transition. The classical individual Bellman equation is

\begin{equation}
v_i(s)
=\max_{a_i\in A_i}
\Big\{\,u_i(s,a_i,a_{-i})
+\gamma\,\mathbb{E}[v_i(s_{T+1})]\Big\}.
\label{bellman}
\end{equation}
where \(s_{T+1}\) denotes the successor state. This formulation underlies both equilibrium reasoning and reinforcement-learning models in which agents adapt through reward accumulation~\cite{Puterman1994}. We use $v_i$ for the classical individual value function; the hybrid coalition value function is denoted $V$ (Section~\ref{neovalue}).
 
In later sections, the continuity and smoothness assumptions implicit in Bellman’s framework are relaxed: utilities may become discontinuous and control alternates binary delegation between Human and AI agents. The corresponding executed-action value function is

\begin{equation}
v(s) =
\mathbb{E}\left[
\sum_{t=0}^{\infty}
\gamma^t
r(s_t, a_t^*)
\;\middle|\;
s_0 = s
\right]
\label{V(s)}
\end{equation}
Here $a_t^*\in A$ denotes the executed action at time $t$, selected by the delegation mechanism. The state evolves according to $s_{t+1}\sim P(\cdot \mid s_t, a_t^*)$ under $\mathcal{V}$.

The proposed actions $a_{H,\;t}\sim p_H(\cdot\mid s_t,t)$ and $a_{AI,\;t}\sim p_{AI}(\cdot\mid s_t,t)$ are internal variables used to compute $D_{JS}$ but they do not enter the transition kernel. The discount factor $\gamma\in(0,1)$ weights future rewards. The resulting regime-conditioned Bellman recursion is formalized in Proposition~\ref{prop:hybridbellman}.

In the present framework, the realized \textbf{coalition reward} equals the utility of the executing agent:

\begin{equation}
r(s_T,a_T^*) :=
\begin{cases}
U_{AI}(s_T,a_T^*) & \text{if } a_T^*=a_{AI},\\
U_H(s_T,a_T^*)    & \text{if } a_T^*=a_H.
\end{cases}
\label{eq:reward_realized}
\end{equation}

The scalar reward function $r:S\times A\to\mathbb{R}$ entering dynamic programming coincides with the coalition utility realized after delegation (defined from the player's utilities \textbf{after} regime selection):

\begin{equation}
\label{rew(s)}
r(s_T,a_T^*;\lambda_T(s_T))
=
\lambda_T(s_T)\,U_H^{T}(s_T,a_T^*)
+
(1-\lambda_T(s_T))\,U_{AI}^{T}(s_T,a_T^*).
\end{equation}

Note that Equation~(\ref{rew(s)}) is equivalent to the piecewise reward definition in Equation~(\ref{eq:reward_realized}); the mixture form is introduced only only for compactness in the Bellman recursion.
At time $T$, the realized reward is $r_T := r_T(s_T,a_T^*) = r(s_T,a_T^*;\lambda_T(s_T))$.
Thus, reward is an instantaneous, regime-conditioned payoff.

Substituting~(\ref{rew(s)}) into~(\ref{V(s)}) yields
\[
V(s)
=
\mathbb{E}\!\left[
\sum_{t=0}^{\infty}
\gamma^t
\big(
\lambda_t(s_t)\,U_H^t(s_t,a_t^*)
+
(1-\lambda_t(s_t))\,U_{AI}^t(s_t,a_t^*)
\big)
\;\middle|\; s_0=s
\right].
\]
Equivalently, the induced Bellman recursion takes the form of Equation (\ref{eq:HybridBellman}) below.

Outside the contextual region, the delegation variable fixes a single executing player.
Within the contextual region, the reward is determined by the utility associated with whichever action is selected by the delegation mechanism.


Throughout this subsection we suppose that Assumption~\ref{ass:lambda} holds, that the discount factor satisfies $\gamma\in(0,1)$, and that the transition kernel $P(\cdot\,|\,s,a)$ is well defined for every executed action.

When evaluating the Bellman operator, the delegation variable $\lambda_T(s)$ is taken as fixed so the maximization is conditional on the realized delegation regime at that step.

\begin{proposition}[Hybrid Bellman recursion]
\label{prop:hybridbellman}
Let $V$ denote the induced value function under the executed delegation regime.
Conditional on the realized delegation variable $\lambda_T(s)$, and given the transition kernel $P(\cdot\mid s,a)$ of $\mathcal{V}$, the Bellman recursion takes the form

\begin{equation}
V(s_T)
=
\max_{a^*\in\mathcal{A}_{\mathrm{exec}}(s)}
\left\{
r(s_T,a^*;\lambda_T(s))
+
\gamma
\mathbb{E}_{s_{T+1}\sim P(\cdot|s,a^*)}[V(s_{T+1})]
\right\},
\label{eq:HybridBellman}
\end{equation}
\noindent
where $\mathcal{A}_{\mathrm{exec}}(s)=\{a_H(s),a_{AI}(s)\}$ denotes the set of feasible executed actions induced by the Human and AI proposals under the delegation mechanism.
\end{proposition}

\noindent

Conditional on the realized delegation regime, the value function $V(s_T)$ satisfies the adapted Bellman recursion in Proposition~\ref{prop:hybridbellman}, once execution authority is selected endogenously.

\begin{proof}
By definition, the realized scalar reward defined in \eqref{rew(s)} is

\[
r(s,a;\lambda_T(s))
=
\lambda_T(s)\,\hat U_H(s,a)
+
(1-\lambda_T(s))\,\hat U_{AI}(s,a),
\]
where the utility traces $\hat U_H^{T}$ and $\hat U_{AI}^{T}$  are the current EWMA estimates. So, when the Bellman operator is evaluated at decision step $T$, the dynamic programming operates on the current traces.
Using the discounted value function,
\[
V(s)
=
\max_{a^*\in\mathcal{A}_{\mathrm{exec}}(s)}
\mathbb{E}\!\left[
r(s,a^*;\lambda_T(s))
+
\gamma V(s_{T+1})
\;\middle|\; s
\right],
\qquad s_{T+1}\sim P(\cdot\mid s,a^*),
\]
the substitution into the recursion yields \eqref{eq:HybridBellman}.
\end{proof}

Equation~\eqref{eq:HybridBellman} extends the classical Bellman recursion \cite{Bellman1957,Puterman1994} by incorporating the delegation factor $\lambda_T(s)$. When $\lambda_T(s)\equiv 1$ or $\lambda_T(s)\equiv 0$, the recursion reduces to the standard single-agent case for $\hat U_H$ or $\hat U_{AI}$. When $\lambda_T(s)$ is state-dependent, the recursion becomes regime-conditioned because the realized reward is selected by the entropy-threshold delegation rule.


\subsection{Neo-Equilibrium}
\label{sec:equilibrium}


\subsubsection{Frequency-convergence equilibrium.}

The equilibrium defined here differs from Nash, Stackelberg, or correlated equilibrium. Classical equilibria characterize fixed points of strategic best responses under static preferences. In contrast, the hybrid equilibrium studied in this model is \textit{dynamical}: it refers to the asymptotic stabilization of the stochastic policy-adaptation process generated by the Robbins–Monro/Markov chain updates and delegation thresholds.

Accordingly, equilibrium is defined in terms of the limiting behavior of the learning dynamics rather than as an instantaneous optimality condition.
Our equilibrium does not describe a one-shot best-response fixed point, but the asymptotic stabilization of the stochastic delegation-learning process generated by repeated interaction under $\mathcal{V}$.

In this framework, the equilibrium is defined through the asymptotic behavior of the empirical regime frequencies. These frequencies satisfy Equation~\eqref{freq_summa}.  

\begin{proposition}     
[Piecewise regular mean-field drift]
\label{lem:piecewise_drift_common}

Let $\mathcal X=\Delta(A_H)\times\Delta(A_{AI})$ denote the compact joint policy space, and let the delegation rule be determined by thresholds on the continuous map $D_{JS}:\mathcal X\to\mathbb R$. For each regime $\rho\in\{\mathrm{agree},\mathrm{context},\mathrm{disagree}\}$, define the one-step update map
$H(x,s_T):=\delta_{a^*(x,s_T)}-x$, and the corresponding conditional mean drift
\[h_\rho(x) :=
\mathbb E\!\left[ H(x,s_T)\mid x_T=x,\ \lambda(x)=\rho\right].\]
Define the overall mean drift by
$h(x)=h_\rho(x)$, whenever x lies in regime $\rho$.

\noindent Assume:
\begin{itemize}
    \item[(i)] The random increment $H(x,s_T)$ is bounded uniformly in $(x,s_T)$;
    \item[(ii)] For each regime $\rho$, the map $h_\rho$ is Lipschitz on the closure of that regime;
    \item[(iii)] $D_{JS}$ is continuous on $\mathcal X$.
\end{itemize}
Then:

\begin{enumerate}
\item The state space $\mathcal X$ is compact and convex.

\item The delegation rule induces a partition of $\mathcal X$ into closed threshold surfaces and open regime interiors.

\item Each regime-wise mean drift $h_\rho$ is bounded and locally Lipschitz on its regime.

\item The overall drift $h$ is piecewise Lipschitz and locally bounded on $\mathcal X$.

\item The stochastic recursion
\[
x_{T+1} = x_T+\eta_T\Big(h(x_T)+M_{T+1}\Big),
\]
where
\[M_{T+1}:= H(x_T,s_T)-h(x_T)\]
has bounded martingale noise and remains in the compact set $\mathcal X$ for every step-size sequence $(\eta_T)$ with $0\le \eta_T\le 1$.
\end{enumerate}
\end{proposition}

\begin{proof}
We proceed in five steps.

\textbf{Step 1: Compactness and convexity of the state space.}
Each simplex $\Delta(A_H)$ and $\Delta(A_{AI})$ is compact and convex in the corresponding Euclidean space. Therefore their product $\mathcal X=\Delta(A_H)\times\Delta(A_{AI})$ is also compact and convex.

\textbf{Step 2: Regime partition induced by the delegation rule.}
By assumption, the delegation rule is determined by threshold comparisons of the continuous map $D_{JS}(x)$. Hence the sets corresponding to the threshold equalities are closed in $\mathcal X$, and the strict-inequality regions are open in the relative topology of $\mathcal X$. Therefore the delegation rule partitions $\mathcal X$ into open regime interiors separated by closed threshold surfaces.

\textbf{Step 3: Regularity of the regime-wise drift.}
Fix a regime $\rho\in\{\mathrm{agree},\mathrm{context},\mathrm{disagree}\}$. By assumption, $h_\rho$ is Lipschitz on the closure of that regime, hence in particular it is locally Lipschitz on the regime interior. Since the closure of each regime is a closed subset of the compact space $\mathcal X$, it is compact. A Lipschitz map on a compact set is bounded. Therefore each $h_\rho$ is bounded on its regime closure.

\textbf{Step 4: Piecewise regularity of the overall drift.}
By definition, the overall drift $h$ agrees with $h_\rho$ on each regime $\rho$. Hence $h$ is Lipschitz within each regime interior. Across threshold surfaces, the formula for the drift may switch from one regime to another, so global Lipschitz continuity need not hold. Nevertheless, because only finitely many regimes are present and each $h_\rho$ is bounded on its regime closure, the overall drift $h$ is locally bounded on $\mathcal X$. Thus $h$ is piecewise Lipschitz and locally bounded.

\textbf{Step 5: Noise decomposition and invariance of the state space.}
Write the recursion in the form
$x_{T+1}=x_T+\eta_T H(x_T,s_T) =x_T+\eta_T\big(h(x_T)+M_{T+1}\big)$, where
$M_{T+1}:=H(x_T,s_T)-h(x_T)$.

By construction, $\mathbb E[M_{T+1}\mid \mathcal F_T]=0$,
so $(M_{T+1})$ is a martingale-difference noise sequence. Since $H$ is uniformly bounded by assumption and $h$ is bounded by Steps 3--4, the noise sequence is also uniformly bounded.

Finally, because $H(x_T,s_T)=\delta_{a^*(x_T,s_T)}-x_T$, the update can be rewritten as
$x_{T+1}=(1-\eta_T)x_T+\eta_T\,\delta_{a^*(x_T,s_T)}$.
For every $0\le \eta_T\le 1$, this is a convex combination of two points in $\mathcal X$: the current policy $x_T\in\mathcal X$ and the pure action profile $\delta_{a^*(x_T,s_T)}\in\mathcal X$. Since $\mathcal X$ is convex, it follows that $x_{T+1}\in\mathcal X$. By induction, all iterates remain in $\mathcal X$.

This proves all claims.
\end{proof}

\begin{remark}[Justification via stochastic approximation]
\label{justification_v_sthicastic}
Under Assumptions~\ref{ass:lambda} (iii) and \ref{eta_assumption}, the coupled policy recursion $(p_H(T),p_{AI}(T))$ is a stochastic approximation scheme with piecewise-Lipschitz drift on a compact state space.
By standard ODE-method results \cite{Benaim1999,KushnerYin2003}, the associated limiting dynamics are governed by the differential inclusion induced by the piecewise-Lipschitz mean-field drift (Proposition~\ref{lem:piecewise_drift_common}).
\end{remark}

In the constant-step regime, ergodicity implies
\[
f_{\mathrm{agree}}(T)
\xrightarrow[T\to\infty]{a.s.}
f_{\mathrm{agree}}^{\ast}
=
\mathbb{E}_{\mu}
\!\left[
\mathbf{1}\!\left\{
D_{JS}^{t} \le \alpha_{\mathrm{agree}}
\right\}
\right],
\]
where $\mu$ is the invariant distribution of the induced Markov chain. In the diminishing-step regime, if $D_{JS}^{t}\to D_{JS}^{\infty}$ almost surely, then by Cesàro convergence,
\[
f_{\mathrm{agree}}(T)
\xrightarrow[T\to\infty]{a.s.}
\mathbf{1}\!\left\{
D_{JS}^{\infty} \le \alpha_{\mathrm{agree}}
\right\}.
\]
Analogous limits hold for $f_{\mathrm{ctx}}(T)$ and $f_{\mathrm{disagree}}(T)$.

\textit{Interpretation.}
If $f_{\mathrm{agree}}^{\ast}=1$, the equilibrium is
AI-dominant (persistent alignment).
If $f_{\mathrm{agree}}^{\ast}=0$, it is Human-dominant (persistent divergence).
Intermediate values correspond to an adaptive equilibrium in which execution regimes stabilize in frequency rather than collapsing to a single regime.

This notion of equilibrium parallels learning-in-games models \cite{Fudenberg1998,Benaim1999,KushnerYin2003}, where repeated interaction leads to stabilization of empirical behavior rather than instantaneous best-response fixed points.

The following proposition formalizes equilibrium as long-run stabilization of the empirical regime frequencies.

\begin{theorem}[Frequency convergence equilibrium]\label{prop:freqconv}
Let
\[
f_{\mathrm{agree}}(T)
:=
\frac{1}{T}\sum_{t=1}^{T}
\mathbf{1}\{D^t_{JS}\le \alpha_{\mathrm{agree}}\}
\]
denote the empirical occupation frequency of the agreement region along the divergence trajectory $\{D^t_{JS}\}_{t\ge1}$, as defined in Equation~\eqref{eq:f_agree}. We distinguish two learning regimes: constant-gain updates yield ergodic time-average convergence, whereas Robbins--Monro step sizes yield almost-sure iterate convergence.

\textit{i) Constant step size.} Assume $\eta_T\equiv \eta>0$ then the joint policy recursion induces a time-homogeneous Markov chain on the compact policy space; convergence then holds in distribution and in ergodic occupation averages.
If this Markov chain is irreducible and aperiodic and admits a unique invariant probability measure $\mu$, then the induced divergence process 
$\{D^t_{JS}\}_{t\ge1}$ satisfies Assumption~\ref{ass:freqconv}, because
$D_{JS}$ is a measurable function of the policy state and thus inherits ergodicity from the policy Markov chain under the same invariant measure $\mu$.

\[
f_{\mathrm{agree}}(T)
\xrightarrow[T\to\infty]{a.s.}
\mathbb{E}_{\mu}
\!\left[
\mathbf{1}\{D_{JS}\le \alpha_{\mathrm{agree}}\}
\right]
\]

\textit{ii) Diminishing step size.} Assume $\{\eta_T\}$ satisfies Assumption~\ref{eta_assumption} (see \cite{Benaim1999,KushnerYin2003}), then the coupled policy recursion is a Robbins--Monro stochastic approximation scheme and its limit set is almost surely contained in the internally chain-transitive invariant set of the associated mean-field ODE \cite{Benaim1999,KushnerYin2003}.
\[
f_{\mathrm{agree}}(T)
\xrightarrow[T\to\infty]{a.s.}
\mathbf{1}\{D^\infty_{JS}\le \alpha_{\mathrm{agree}}\}.
\]
\end{theorem}

\begin{proof}
(i) Under constant step size, the policy recursion defines
a time-homogeneous Markov chain on a compact state space.
By Assumption~\ref{ass:freqconv}, this chain admits a unique invariant distribution $\mu$.
Since the indicator is bounded, the ergodic theorem yields
\[
\frac{1}{T}\sum_{t=1}^{T}
\mathbf{1}\{D^t_{JS}\le \alpha_{\mathrm{agree}}\}
\xrightarrow[T\to\infty]{a.s.}
\mathbb{E}_{\mu}
\!\left[
\mathbf{1}\{D_{JS}\le \alpha_{\mathrm{agree}}\}
\right].
\]

(ii) Under Assumption~\ref{eta_assumption}, the recursion is a standard Robbins--Monro stochastic approximation scheme.
If $D^t_{JS}\to D^\infty_{JS}$ almost surely, then by continuity of the indicator away from the threshold,
\[
\mathbf{1}\{D^t_{JS}\le \alpha_{\mathrm{agree}}\}
\xrightarrow{a.s.}
\mathbf{1}\{D^\infty_{JS}\le \alpha_{\mathrm{agree}}\}.
\]
Therefore Ces\`aro convergence implies
\[
f_{\mathrm{agree}}(T)
\to
\mathbf{1}\{D^\infty_{JS}\le \alpha_{\mathrm{agree}}\}
\quad a.s.
\]
\end{proof}

If $f_{\mathrm{agree}}(s)=1$, the equilibrium is \emph{AI-dominant}, meaning that policies remain within the agreement band. If $f_{\mathrm{agree}}(s)=0$, it is \emph{Human-dominant}, corresponding to persistent divergence. Intermediate values define an \emph{adaptive} equilibrium. This interpretation parallels learning-in-games models \cite{Fudenberg1998,Benaim1999,KushnerYin2003}. Theorem~\ref{prop:freqconv} should therefore be read as a dynamic stabilization result for empirical regime frequencies, not as a uniqueness theorem or a strategic optimality result \cite{KushnerYin2003,MeynTweedie2009,BenaimHofbauerSorin2005,Borkar2008}.

\subsubsection{Divergence equilibrium}

At the divergence level, equilibrium is inherited from policy stabilization. Since $D_{JS}^{T}=D_{JS}\!\big(p_H(a\mid s_T,T),p_{AI}(a\mid s_T,T)\big)$ is a function of the empirical policies, asymptotic stabilization of the policy process induces stabilization of the divergence sequence.

\begin{assumption}[Statistical regularity of the divergence process]
\label{ass:freqconv}
The divergence sequence $\{D_{JS}^{t}\}_{t\ge1}$ generated by the policy updates satisfies the following asymptotic regularity conditions:

\begin{enumerate}[label=(\roman*),itemsep=0.2em,topsep=0.1em]
\item The process $\{D_{JS}^{t}\}_{t\ge1}$ is asymptotically stationary and ergodic, with finite variance.
\item The indicator variable $\mathbf{1}\{D_{JS}^{t}\le\alpha_{\mathrm{agree}}\}$
is bounded and integrable.
\end{enumerate}
\end{assumption}

Assumption~\ref{ass:freqconv} is stated at the level of the divergence process $\{D^t_{JS}\}_{t\ge1}$ and applies to justify frequency convergence in the constant step-size regime considered in Theorem~\ref{prop:freqconv}, where stabilization occurs in distribution rather than almost surely. Under constant step size, irreducibility and aperiodicity of the joint policy Markov chain yield asymptotic stationarity and ergodicity of the induced divergence process. Under diminishing step sizes, the policy process is nonstationary during the transient phase, but the induced divergence sequence may satisfy the same asymptotic regularity once $\eta_T$ vanishes and the policies stabilize.


\section{\label{sec:FirstSC} Scenario 1: Human arbitration.}

This section develops the first regime of Neo-Game Theory: Human arbitration. We specify the delegation mechanism for this scenario and describe how it is implemented in the simulations, including the contextual rule $\lambda_{\mathrm{ctx}}$ used in the intermediate delegation region.


\subsection{Description}

 The regime is termed \emph{Human arbitration} because the Human retains ultimate authority over execution. Full disagreement ($D_{JS} \ge \alpha_{\mathrm{disagree}}$) enforces Human execution whereas AI autonomy occurs only under sufficient policy alignment ($D_{JS} \le \alpha_{\mathrm{agree}}$). 
 
 Within the contextual region ($\alpha_{\mathrm{agree}} < D_{JS} < \alpha_{\mathrm{disagree}}$), execution is governed by the simulated contextual rule defined in Equation (\ref{lambda_context}). This simulation-based mechanism models the intermediate zone where real Human choice cannot be specified \emph{ex-ante}.
Delegation is calibrated through empirical policy agreement. This scenario models the \emph{teacher-forcing} stage of hybrid learning and is designed to capture:
\begin{itemize}
    \item \textbf{Asymmetric trust and authority:} the AI adapts while the Human retains structural override authority.
    
    \item \textbf{Reward sensitivity to alignment:} improvement is measured through reduced $D_{JS}$.
   
     \item \textbf{Adaptive delegation:} control transitions across agreement, contextual, and disagreement regimes based on empirical policy divergence.

\end{itemize}
The architecture supports human-in-the-loop training, early-deployment adaptation, and trust-sensitive systems.

\subsection{Delegation mechanism}
\label{delegation_sc1}

We refer to the resulting delegation contexts (agreement, contextual, and disagreement) as execution regimes, since each determines how the delegation variable $\lambda_T(s)$ is selected and therefore which agent executes the action at time $T$.

\begin{definition}[Two-threshold entropy-based delegation]
\label{twothreshold}
Let $p_H(\cdot \mid s,T)$ and $p_{AI}(\cdot \mid s,T)$ denote the Human and AI policies at state $s$ and time $T$, defining the Jensen--Shannon divergence at instant $T$ as in Equation(\ref{D_JS_T})

Let two thresholds satisfy $0 \le \alpha_{\mathrm{agree}} < \alpha_{\mathrm{disagree}} \le 1$. The delegation variable $\lambda_T(s) \in \{0,1\}$ is the scenario-specific form of the general delegation rule in Equation~(\ref{eq:lambda_rule}):

\begin{equation}
\label{eq:lambda_case1}
\lambda_T(s) =
\begin{cases}
0 & \text{if}\; D_{JS}^T  \le \alpha_{\mathrm{agree}},\\[4pt]
\lambda_{\mathrm{ctx}}(D^T_{JS}) & \text{if}\; \alpha_{\mathrm{agree}} < D_{JS}^T  < \alpha_{\mathrm{disagree}},\\[4pt]
1 & \text{if}\; D_{JS}^T  \ge \alpha_{\mathrm{disagree}},
\end{cases}
\end{equation}
where $\lambda_{ctx}$ is $\lambda_T(s)$ in the contextual zone, defined in Equation~\eqref{lambda_context}. 
\end{definition}

In the contextual region described below, execution is determined by $\lambda_{\mathrm{ctx}}(D_{JS}^T)$. Since real Human choice is not observable \emph{ex ante} in this intermediate zone, the contextual decision is modeled probabilistically for implementation purposes.

The delegation mechanism requires a bounded and continuous discrepancy statistic in order to partition the policy space into agreement, contextual, and disagreement regions. We use the Jensen--Shannon divergence which makes the thresholds directly interpretable (Section~\ref{sec:Entropy}).

\emph{Coalition preference is hierarchical}: the threshold rule selects the execution regime before utilities are compared within that regime. Consequently, no compensatory trade-off exists between outcomes belonging to different execution regimes. This induces a non-Archimedean (lexicographic) structure (Section~\ref{non_archi}) at the coalition level, even though $U_H$ and $U_{AI}$ are individually real-valued. 

Outside the contextual region, the realized utility coincides with the corresponding executor’s utility. Thus, Equation~(\ref{eq:reward_realized}) provides the realized within-regime utility, while cross-regime comparison follows the lexicographic rule in Equation~\eqref{eq:lex_order}.

Within the contextual region, $D_{JS}^T$ is interpreted as the probability of Human execution, while $1-D_{JS}^T$ is the probability of AI execution. Thus higher divergence increases the likelihood of Human intervention. Let $\kappa\sim\mathrm{Uniform}(0,1)$; then

\begin{equation}
\label{lambda_context}  
\lambda_{\mathrm{ctx}} =
\begin{cases}
1 & \text{if } \kappa < D_{JS}^T,\\
0 & \text{otherwise}.
\end{cases}
\end{equation}
i.e. $\Pr(\lambda = 1) = \Pr(\kappa < D_{JS}^T) = D_{JS}^T$
and $\Pr(\lambda = 0) = 1 - D_{JS}^T$. 


\subsection{Implementation}
\label{sec:Implementation}

This subsection describes the implementation of Scenario~1 whose base-case pseudocode is reported in Appendix~\ref{pseudocode}.

\vspace{0.2cm}

a) When \texttt{random\_states=False} (\texttt{rs=F}), state transitions follow the action-dependent kernel $P(\cdot\mid s_t,a_t^\ast)$ described in Section~\ref{virtualnature}. This is the Markov regime: the environment is governed by the transition kernel rather than by an exogenous state distribution. In this case, no single state-independent probability $p_A$ exists, because the probability of reaching state $A$ depends jointly on the current state $s_t$ and the executed action $a_t^\ast$. 

\vspace{0.2cm}

b) When \texttt{random\_states=True} (\texttt{rs=T}), $\mathcal{V}$ generates states exogenously, independently of the executed action. This specification concerns only state generation. The learning regime is determined separately by the step-size schedule: constant $\eta$ yields a constant-gain update, whereas diminishing $\eta_T$ yields a Robbins--Monro update where
\[
\Pr(s_{t+1}=A)=p_A, \;\; \Pr(s_{t+1}=B)=1-p_A.
\]
These runs therefore evaluate the delegation--learning dynamics under controlled state frequencies and do \emph{not} use an action-dependent transition kernel $P(\cdot\mid s,a)$.

\vspace{0.2cm}
c) The \textbf{realised coalition reward} depends on the executed action $a_T^\ast$:
\[
r_T^{\text{coalition}}
=
\mathbf{1}\{\lambda_T = 1\}\,\hat{U}_H(s_T,a_T^\ast)
+
\mathbf{1}\{\lambda_T = 0\}\,\hat{U}_{AI}(s_T,a_T^\ast).
\]
This is the empirical reward used in evaluation, as in Equation~\eqref{eq:reward_realized}.

For value computation, the Bellman backup evaluates all candidate actions through the delegation-conditioned stage reward:

\[
\hat V(s_T)
=
\max_{a\in A}
\left\{
r_T(s_T,a)
+
\gamma
\sum_{s'} P_V(s_{T+1}|s_T,a)\,V(s_{T+1})
\right\},
\]
where the stage reward is
$r_T(s_T,a)=\mathbf{1}\{\lambda_T = 1\}\,\hat{U}_H(s_T,a)
+\mathbf{1}\{\lambda_T = 0\}\,\hat{U}_{AI}(s_T,a)$.

\vspace{0.3cm}

Table~\ref{Table1} reports the initialization variables. To avoid notation ambiguity, we distinguish the structural utilities from their online estimates: \(U_H(s,a)\) and \(U_{AI}(s,a)\) denote the underlying Human and AI utilities, whereas \(\hat U_H(s,a)\) and \(\hat U_{AI}(s,a)\) denote the utility traces updated during learning. Accordingly, the online traces are initialized as \(\hat U_H(s,a)=0\) and \(\hat U_{AI}(s,a)=0\) for all state-action pairs.  If not provided, the initial Human and AI policy distributions (marked with an asterisk in Table~\ref{Table1}) default to the symmetric distribution initialized as uniform (Dirichlet(1) smoothing).

\begin{table}[h]
\centering
\captionsetup{labelfont={color=black},textfont={color=black}}
\caption{Initialization variables used in the simulations.}

\begin{tabular}{ll}
\hline
Human initial utility trace & $\hat U_H(s,a) = 0$ \\
AI initial utility trace & $\hat U_{AI}(s,a) = 0$ \\
Initial values & $V(A) = V(B) = 0$ \\
Human policy & $p_H(\cdot \mid s)$ given* \\
AI policy & $p_{AI}(\cdot \mid s)$ given* \\
Initial update frequencies & $f_H(0) = f_{AI}(0) = 0$ \\
Initial states & $s_0 \sim (p_A,p_B)$ or $s_0 = A$ \\
Actions set & $\{a_0,a_1\}$ \\
\hline
\end{tabular}
\label{Table1}
\end{table}
\noindent

\medskip


\textbf{\textit{Inputs.}} Each grid configuration is specified by the delegation thresholds $\alpha_{\mathrm{agree}}$ and $\alpha_{\mathrm{disagree}}$, the step size $\eta$ or $\eta_{\mathrm{decay}}$, the EWMA utility rate $\beta$, and the contextual rule. Outside the contextual region, simulations use deterministic two-threshold delegation as in Definition~\ref{twothreshold}. Within the contextual region, selection is simulated probabilistically.

To assess convergence properties, the algorithm is executed separately for each horizon \(H\in\{200,1000,5000,10000,15000,30000,50000\}\), rather than pooling all horizons within a single grid search. For each fixed \(H\), the full parameter grid reported in Table~\ref{tab:param_grid} is evaluated independently, and the best-performing configuration is selected according to the joint score criterion. Figure~\ref{fig:grid} shows the representative case \(H=5000\), while Table~\ref{tab:param_grid} summarizes the full set of horizons and parameter values used in the experiments.

\begin{table}[H]
\footnotesize
\captionsetup{labelfont={color=black},textfont={color=black}}
\caption{Inputs parameter grid used in the simulations.
$0 \le \alpha_{\mathrm{agree}} < \alpha_{\mathrm{disagree}} \le 1$.}
\label{tab:param_grid}
\centering
\footnotesize
\arrayrulecolor{black}
\setlength{\tabcolsep}{6pt}
\renewcommand{\arraystretch}{1.15}
\begin{tabularx}{\textwidth}{>{\color{black}}l >{\color{black}}X}
\toprule
\textbf{Parameter} & \textbf{Values} \\
\midrule
$H$ & 200, 1000, 5000, 10000, 15000, 30000, 50000 \\
$\alpha_{\mathrm{agree}}$ & 0.20, 0.35, 0.49, 0.70 \\
$\alpha_{\mathrm{disagree}}$ & 0.35, 0.51, 0.85 \\
$\beta$ & 0.02, 0.05, 0.10 \\
$\gamma$ & 0.1, 0.25, 0.5, 0.75, 0.90 \\
\texttt{seed} & 42 \\
\texttt{random\_states} & True, False \\
\texttt{eta\_label} & \texttt{const\_0p05}, \texttt{decay\_0p05\_over\_1p0plus\_0p001t} \\
\texttt{state\_probs} &
(0.7, 0.3), (0.9, 0.1), (0.01, 0.99),\\
& (0.3, 0.7), (0.1, 0.9), (0.99, 0.01) \\
\bottomrule
\end{tabularx}
\end{table}
\FloatBarrier

\textbf{\textit{Outputs.}}
For each configuration, the simulator records:
(i) Jensen--Shannon divergence $D_{\mathrm{JS}}$ (mean and terminal);
(ii) execution frequencies $f_H$, $f_{AI}$, and $f_{\mathrm{overall}}$;
(iii) regime frequencies $(f_{\mathrm{agree}}, f_{\mathrm{ctx}}, f_{\mathrm{disagree}})$;
(iv) ownership metrics (overall and contextual);
(v) terminal values $(D_{\mathrm{JS}}^{\mathrm{final}}, f_{\mathrm{overall}}^{\mathrm{final}})$.

Certain parameter configurations generate irregular or mixed policy trajectories. Because delegation is governed by entropy thresholds and contextual arbitration, the executed action may alternate stochastically between Human and AI control. When this occurs frequently---especially under large learning rates $\eta$ or wide contextual regions---the empirical policies may temporarily diverge or oscillate before eventual stabilization. Such patterns reflect the stochastic delegation dynamics of the hybrid system rather than a failure of the  framework.
  
\subsection{Results}
\label{subsec Results}

We report the simulation results in two steps. First, we isolate the role of the main parameters in shaping transient learning behavior. Second, we summarize the cross-horizon regularities that characterize the asymptotic behavior of the Human arbitration regime.

\subsubsection{\label{sec:Analysis} Parameter analysis}
\label{parameters-analysis}

\textit{(a) \textbf{Delegation thresholds.}} The thresholds $\alpha_{\mathrm{agree}}$ and $\alpha_{\mathrm{disagree}}$ determine how often the process enters agreement, contextual, or disagreement regimes. The learning schedule controls the trade-off between rapid adjustment and oscillatory behavior, while the environmental specification (\texttt{r\_s=True}/\texttt{False}) shapes the short-run variability of state visitation. The exogenous state-sampling pair $(p_A,p_B)$ matters only when \texttt{r\_s=True}; when \texttt{r\_s=False}, state evolution is generated endogenously by $P(\cdot\mid s_T,a_T^\ast)$.  These factors affect transient trajectories but not the long-run alignment pattern. The supplementary material reports \textcolor{blue}{(\href{https://github.com/sepulveda-fontaine-s/Neo_Game_Theory/tree/main}{GitHub}, (accessed on 10 April 2026) \cite{SepulvSupplMat}) }the best-performing configurations for each horizon and shows that the same qualitative interpretation holds across the full grid.

\textit{(b) \textbf{Policy–update rate/schedule}, $\eta$.} The learning rate $\eta$ governs how strongly each executed action updates the policy through the Robbins–Monro rule
$(1-\eta)p+\eta,\delta_{a^*}$.
Large values produce rapid but potentially volatile policy changes, whereas smaller values yield slower but smoother adaptation.
A constant step size leads to fast early alignment but may generate small oscillations because updates never vanish. In contrast, the decaying schedule $\eta_T = 0.05/(1+0.001T)$ satisfies Robbins–Monro conditions and gradually reduces update magnitudes, producing smoother convergence trajectories. Across horizons, constant learning rates accelerate early reductions in $D_{JS}$, while decaying schedules provide greater asymptotic stability.
At sufficiently large horizons, however, both schedules lead to similar equilibrium policies, indicating that $\eta$ mainly governs the \emph{trajectory} of learning rather than the final outcome.

\textit{(c) \textbf{Action-dependent Virtual Nature experiment.}} The environment can evolve either exogenously (\texttt{$r\_s=True$}) or through an action-dependent transition kernel (\texttt{$r\_s=False$}). 
Nevertheless, both regimes ultimately produce similar equilibrium policies once the horizon is sufficiently long.

We retain the two-state structure $S=\{A,B\}$ and 
binary action set $A=\{0,1\}$. 
With complementary probabilities corresponding to remaining in the same state. In this configuration, state visitation is endogenously shaped by the executed action $a_T^*$. Therefore, this experiment evaluates whether the entropy-based delegation dynamics remain qualitatively consistent when Virtual Nature is fully action-dependent.

\textit{(d) \textbf{Utility–credit rate}, $\beta$.} It determines the responsiveness of utility traces through the EWMA update rule.
Higher values increase the speed with which utilities react to executed actions, amplifying early adjustments in delegation frequencies.
Lower values produce smoother but slower adaptation. Across horizons, the influence of $\beta$ is mainly visible in the early stages of learning: larger values accelerate transient adjustments in $D_{JS}^T$ and execution frequencies, while smaller values produce more gradual trajectories. However, the long-run policy configuration remains largely unaffected by $\beta$, indicating that the parameter primarily controls the \emph{timescale} of utility adaptation.

\subsubsection{Cross-iteration synthesis.} 

A summary of results after implementation is shown in the following items. Detailed information and plots about each iteration can be found in
\href{https://github.com/sepulveda-fontaine-s/Neo_Game_Theory/tree/main}{GitHub} \cite{SepulvSupplMat}.

All reported counts aggregate outcomes across the full parameter grid (Table~\ref{tab:param_grid}).

\begin{itemize}

\item $H=200$: \textbf{Early exploratory regime}.
At this short horizon the system remains dominated by transient dynamics. Human and AI policies typically remain misaligned and $D_{JS}$ seldom converges to zero.  When \texttt{r\_s=True}, stochastic state sampling produces higher volatility and slower reductions in divergence, whereas deterministic transitions  (\texttt{r\_s=False}) yield smoother trajectories and slightly faster alignment. Constant $\eta$ generates larger oscillations in policy updates, while decaying $\eta(t)$ produces more stable learning paths but still does not guarantee convergence.

\item $H=1{,}000$: \textbf{Early learning}. 
Initial alignment begins to appear in several configurations. 
With constant $\eta$, divergence often decreases quickly but residual oscillations remain, especially under \texttt{r\_s=True}. 
The Markov environment (\texttt{r\_s=False}) typically reduces $D_{JS}$ faster because state visitation is structured by actions. 
Under decaying $\eta(t)$ trajectories become smoother and more stable, though convergence across both states is not yet systematic.

\item $H=5{,}000$: \textbf{Emerging stabilisation}. 
Most configurations show clear reductions in $D_{JS}$ and increasing similarity between Human and AI policies. Constant $\eta$ produces faster early convergence but may still generate mild oscillations in the stochastic environment. Decaying $\eta(t)$ yields smoother trajectories and more stable alignment across both states. Differences between \texttt{r\_s=True} and \texttt{False} remain visible but become less pronounced. 

\item $H=10{,}000$: \textbf{Near-equilibrium}. 
Most configurations now reach practical convergence. With constant $\eta$, policies generally align across states but small oscillations may persist due to the fixed update step, particularly when \texttt{r\_s=True}, where stochastic state sampling continues to introduce minor fluctuations. Under \texttt{r\_s=False}, the action-dependent transition structure produces smoother trajectories and faster collapse of $D_{JS}$. 
Decaying $\eta(t)$ stabilises the learning dynamics in both environments, yielding near-flat trajectories and consistent policy alignment.

\item $H=15{,}000$: \textbf{Stabilised regime}. 
Human and AI policies coincide across most configurations and $D_{JS}$ is effectively zero. Constant $\eta$ still produces small fluctuations in policy probabilities, particularly in the stochastic environment (\texttt{r\_s=True}), although these do not alter the overall alignment outcome. With \texttt{r\_s=False}, trajectories remain smoother because state visitation is structured by the action-dependent transition kernel. 
Decaying $\eta(t)$ further dampens updates, producing almost stationary trajectories.

\item $H=30{,}000$: \textbf{Robbins--Monro freeze-out}. 
At this horizon all four regimes approach stable policy configurations. 
With constant $\eta$, both \texttt{r\_s=True} and \texttt{r\_s=False} reach strong policy alignment across states, although very small oscillations in policy probabilities may persist because the learning step does not vanish. 
Under decaying $\eta(t)$ the Robbins--Monro schedule effectively freezes the updates, producing nearly flat trajectories in both stochastic and Markov environments. In all cases the $D_{JS}$  reaches almost zero, indicating practical convergence of Human and AI policies across states.

\item $H=50{,}000$: \textbf{Saturation}. 
The system reaches a saturated regime in which both learning schedules produce nearly identical outcomes. With decaying $\eta(t)$ trajectories remain flat at their limiting values, while constant $\eta$ may generate extremely small residual fluctuations. The distinction between \texttt{r\_s=True} and \texttt{r\_s=False} becomes negligible, as policies are fully aligned and $D_{JS}$ remains effectively zero across states.

\end{itemize}
\noindent
Figure~\ref{fig:grid} illustrates the dynamics at this horizon, where the learning process is clearly visible before full asymptotic saturation. The trajectories show the contraction of $D_{JS}$ together with the convergence of the Human and AI policies $p_H$ and $p_{AI}$.

\begin{figure}[H]
\centering

\subfloat[$D_{\mathrm{JS}}$ rapidly decreases toward zero, indicating fast alignment between Human and AI policies. Human components stabilize early, while the AI policy displays a longer transient adjustment before convergence. Constant learning yields small residual fluctuations but does not prevent stabilization. $\alpha_{\mathrm{agree}}=0.2$, $\alpha_{\mathrm{disagree}}=0.85$, $\beta=0.02$, $(p_A,p_B)$ not applicable, $\gamma=0.1$.]{%
\begin{minipage}[b]{0.49\linewidth}
\centering
$D_{\mathrm{JS}}$ \& Policies with r\_s=False, $\eta_c=0.05$\par\smallskip
\includegraphics[width=\linewidth]{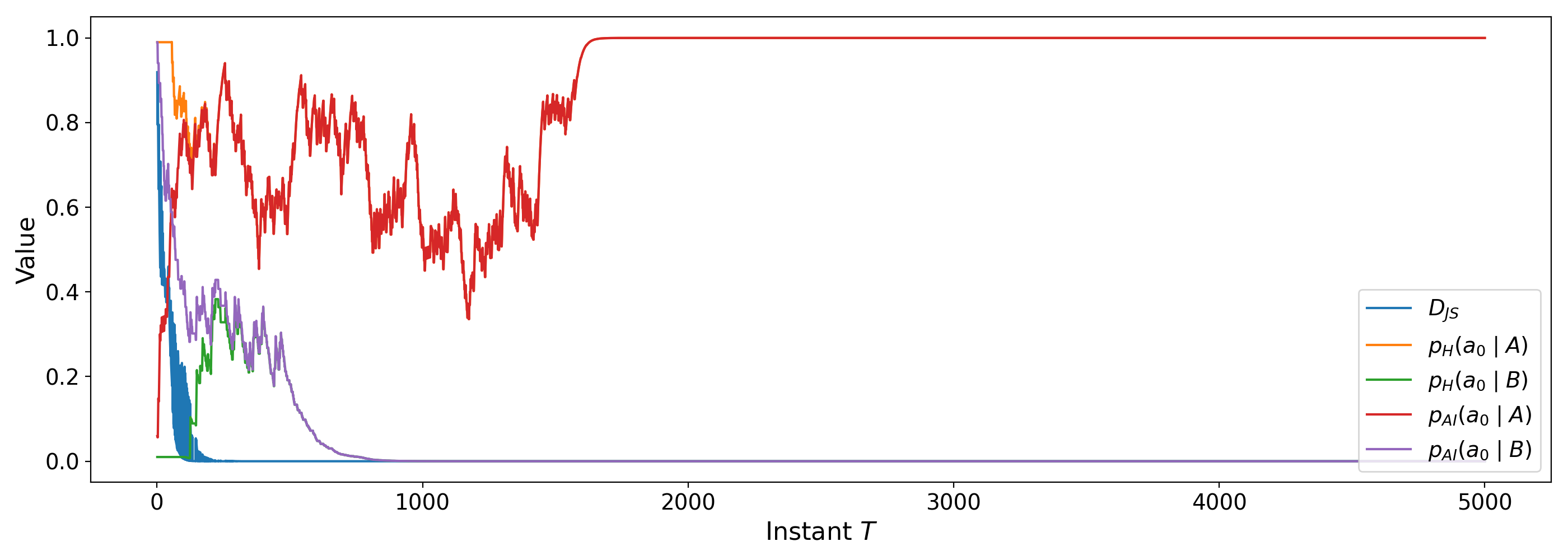}
\end{minipage}}
\hfill
\subfloat[$D_{\mathrm{JS}}$ declines rapidly in the early iterations, indicating fast alignment between Human and AI policies. Stochastic state realizations and the constant learning rate generate visible transient fluctuations, but the policies still converge across states. $\alpha_{\mathrm{agree}}=0.2$, $\alpha_{\mathrm{disagree}}=0.35$, $\eta=0.05$, $\beta=0.02$, $p_A=0.7$, $p_B=0.30$, $\gamma=0.1$.]{%
\begin{minipage}[b]{0.49\linewidth}
\centering
$D_{\mathrm{JS}}$ \& Policies with r\_s=True, $\eta_c=0.05$\par\smallskip
\includegraphics[width=\linewidth]{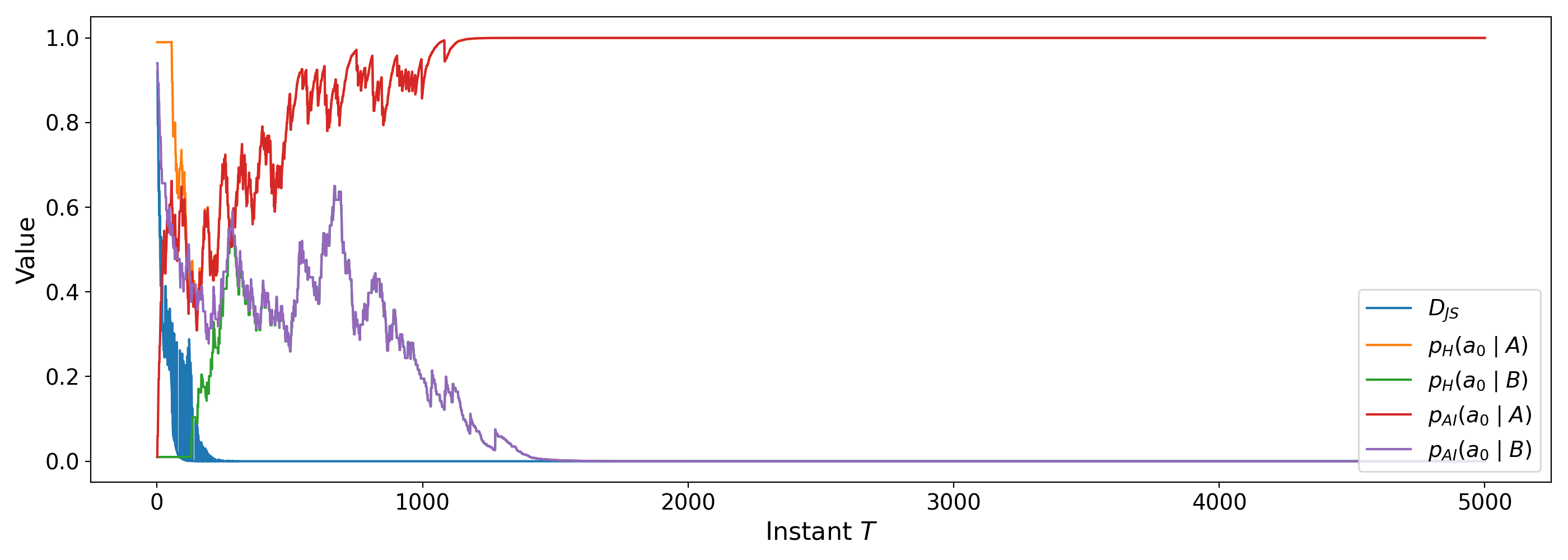}
\end{minipage}}
\par\medskip
\subfloat[The Robbins--Monro schedule dampens early oscillations while $D_{\mathrm{JS}}$ collapses rapidly, indicating alignment of Human and AI policies under action-dependent transitions. $\alpha_{\mathrm{agree}}=0.35$, $\alpha_{\mathrm{disagree}}=0.51$, $\eta(t)=0.05/(1+0.001t)$, $\beta=0.02$, $(p_A,p_B)$ not applicable, $\gamma=0.1$.]{%
\begin{minipage}[b]{0.49\linewidth}
\centering
$D_{\mathrm{JS}}$ \& Policies with r\_s=False, $\eta(t)=0.05/(1+0.001t)$\par\smallskip
\includegraphics[width=\linewidth]{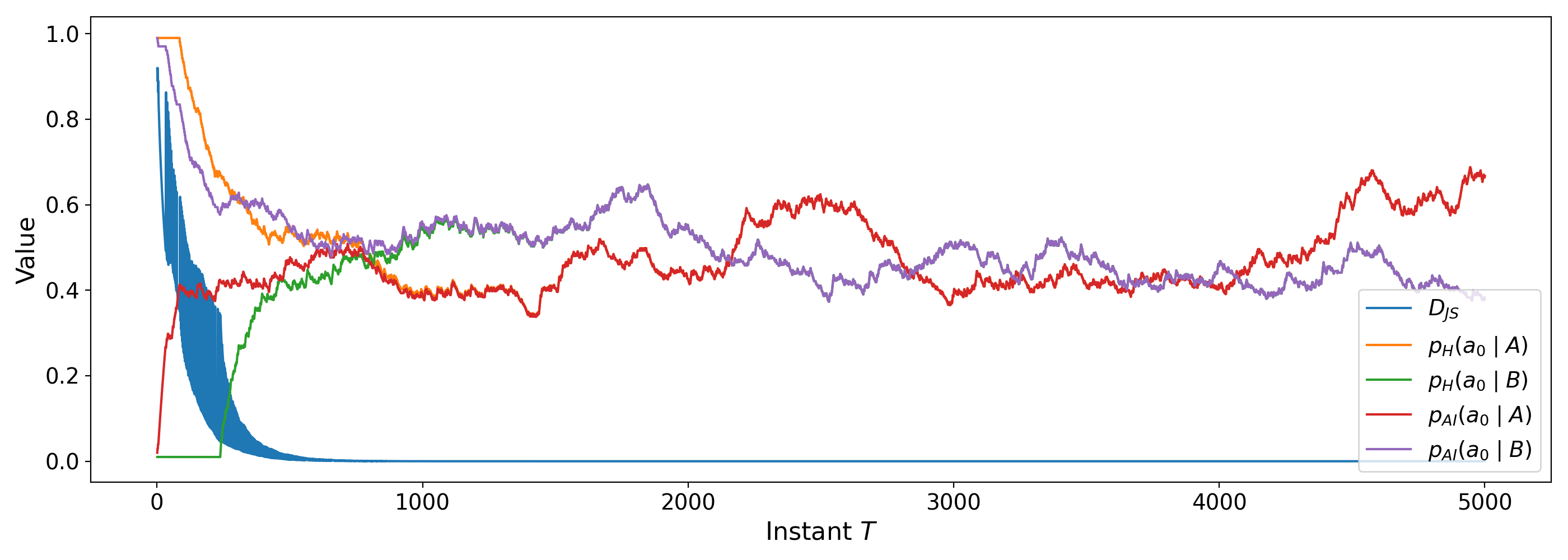}
\end{minipage}}
\hfill
\subfloat[The Robbins--Monro schedule reduces early oscillations while $D_{\mathrm{JS}}$ collapses rapidly, indicating alignment of Human and AI policies under stochastic state generation. $\alpha_{\mathrm{agree}}=0.35$, $\alpha_{\mathrm{disagree}}=0.51$, $\beta=0.02$, $p_A=0.30$, $p_B=0.70$, $\gamma=0.1$.]{%
\begin{minipage}[b]{0.49\linewidth}
\centering
$D_{\mathrm{JS}}$ \& Policies with r\_s=True, $\eta(t)=0.05/(1+0.001t)$\par\smallskip
\includegraphics[width=\linewidth]{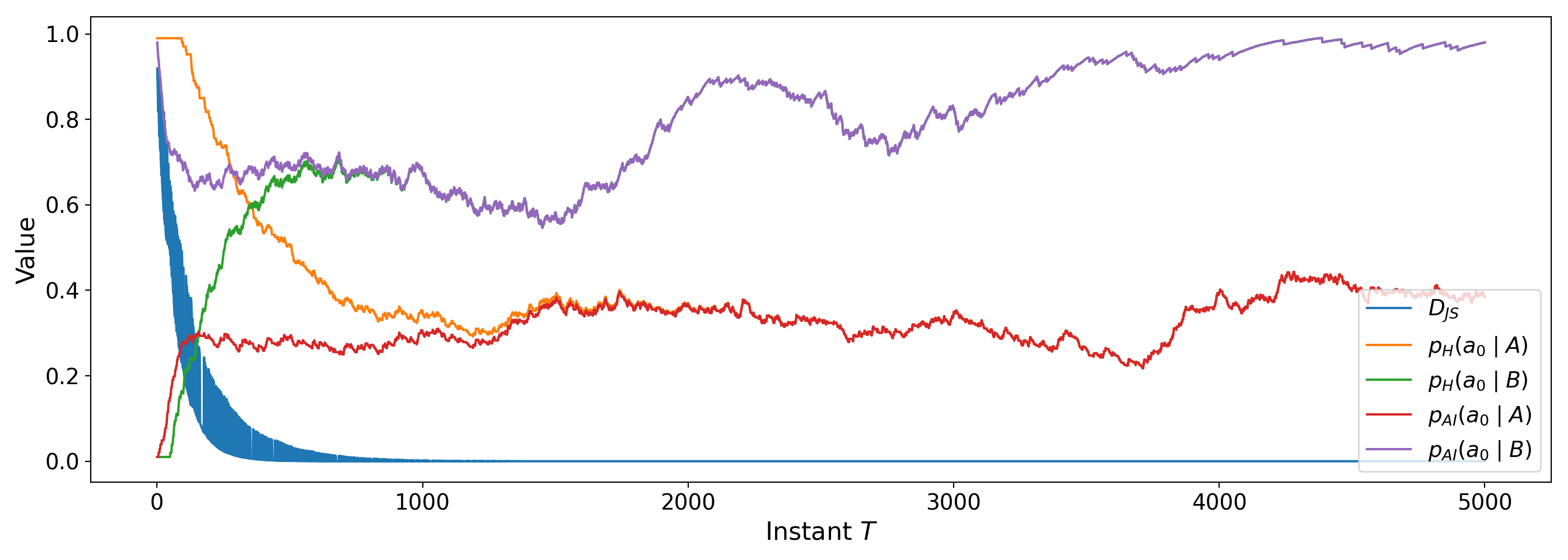}
\end{minipage}}

\caption{Each panel shows the time evolution of the Jensen--Shannon divergence $D_{\mathrm{JS}}$ together with the state-conditional policy probabilities $p_H(a_0\mid s)$ and $p_{AI}(a_0\mid s)$. The trajectories correspond to representative runs with minimal final divergence within each regime. Across all configurations, $D_{\mathrm{JS}}$ contracts rapidly and the Human and AI policies progressively align, illustrating the robustness of the delegation dynamics under different learning schedules and environmental specifications.}
\label{fig:grid}
\end{figure}
\FloatBarrier

\section{Discussion}
\label{sec:Discussion}

This work establishes the operational foundations of Neo-Game Theory and the Human arbitration regime, the Scenario 1 discussed in Section \ref{sec:FirstSC}. In this scenario, Human retains structural authority, while the AI adapts through observation, empirical frequency tracking, and repeated exposure to executed actions under $\mathcal{V}$.
This stable Human--AI alignment can emerge without classical best-response optimization. The result is a dynamic form of stabilization: Jensen--Shannon divergence $D_{JS}$ contracts, the policies align, and the occupation frequencies of the delegation regimes become stable over time.

The simulations also show that the principal role of the parameters is transient rather than asymptotic. Indeed, constant learning rates produce faster but noisier adjustment; decaying learning rates reduce oscillations; exogenous state generation increases early variability; and action-dependent transitions produce more structured short-run paths. Yet, in all these configurations, the long-run pattern remains robust: the Human and AI policies converge toward alignment and the divergence approaches zero. Across all horizons, the main parameters affect the speed and stability of learning more than the existence of convergence itself. At short horizons, the learning schedule and environmental specification produce visibly different trajectories: stochastic state generation increases volatility, whereas action-dependent transitions yield smoother paths; constant $\eta$ accelerates adaptation but may induce oscillations, whereas Robbins--Monro decay stabilizes the trajectories.
As the horizon increases, these differences progressively diminish. Under both environmental regimes and both learning schedules, Human and AI policies align across states and $D_{JS}$ approaches zero. Decaying step sizes suppress residual fluctuations more effectively, while constant step sizes may preserve small oscillations without altering the limiting pattern.

The aforementioned distinction helps situate the Human--AI arbitration regime within the broader logic of Neo-Game Theory and highlights the theoretical importance of its transient dynamics.

\section{Conclusion}
\label{sec:Conclusion}

This paper introduced Neo-Game Theory as a framework for hybrid Human--AI coalitions operating under Virtual Nature $\mathcal{V}$. In the Human arbitration regime, the only one studied in the present work, delegation is governed by the Jensen--Shannon divergence $D_{JS}$, the thresholds $\alpha_{\mathrm{agree}}$ and $\alpha_{\mathrm{disagree}}$, and the delegation factor $\lambda_T(s)$. The resulting equilibrium concept is dynamic and frequency-based rather than a static best-response solution.

The numerical results reported in Section \ref{subsec Results} confirm that convergence in the Human arbitration regime does not rely on classical best-response optimization. Instead, alignment emerges from repeated execution, empirical adaptation, and the statistical structure of the delegation rule. Therefore, stability has a behavioral interpretation: Human and AI policies remain aligned and the divergence remains near zero over long horizons.

Overall, our results indicate that the \textbf{entropy-based delegation mechanism robustly drives policy alignment}: parameter choices shape transient learning paths but do not materially alter the long-run outcome once sufficient iterations are allowed.

Future work will develop the remaining regimes of Neo-Game Theory, namely, AI-control and Negotiation. These extensions will allow systematic comparisons with established solution concepts such as Nash, Stackelberg, and correlated equilibrium, with the aim of developing a unified taxonomy of Human–AI equilibria in digitally mediated environments governed by Virtual Nature.

\vspace{6pt} 




\textbf{Data Availability Statement:} No new data were created or analyzed in this study.


\textbf{Conflicts of Interest:} The authors declare no conflicts of interest.

\clearpage

\appendix

\section{Glossary}
\label{Glossary}

The table below summarizes the main notation used in this work.

\begin{table}[h!]
\centering
\caption{Notation and key variables in the hybrid Human--AI equilibrium framework.}
\begin{tabular}{ll}
\toprule
\textbf{Symbol} & \textbf{Description} \\
\midrule
$s$ & Environment or system state at a given time step \\
$s_T$ & Realized state of the system at decision time $T$ \\
$a_{H,T},\,a_{AI,T}$ & Actions proposed by the Human and the AI at time $T$ \\
$a_T^\ast$ & Executed action after the delegation rule (Definition~\ref{twothreshold}) \\
$A_H,\,A_{AI}$ & Action sets of the Human and AI agents \\
$A$ & Executed action space, with $A = A_H \cup A_{AI}$ \\
$U_H(s,a),\,U_{AI}(s,a)$ & Structural utility functions of the Human and AI agents \\
$\widehat U_H^{T}(s,a),\,\widehat U_{AI}^{T}(s,a)$ & Online utility estimates (EWMA traces) used for learning updates \\
$U_{\mathrm{joint}}(\lambda_T(s),s_T,a_T^\ast)$ & Coalition utility induced by the executed action (Equation~\eqref{eq:Ujoint}) \\
$r(s,a;\lambda_T(s))$ & \parbox[t]{0.55\textwidth}{Scalar reward entering the Bellman recursion, determined by the executing agent (Equation~\eqref{rew(s)})} \\
$\lambda_T(s)$ & Delegation selector at time $T$ indicating which agent executes the action \\
$\alpha_{\mathrm{agree}}$ & Lower Jensen--Shannon divergence threshold defining the agreement regime \\
$\alpha_{\mathrm{disagree}}$ & Upper divergence threshold defining the disagreement regime \\
$\beta$ & EWMA learning rate used in the utility update rule \\
$\eta_T$ & Robbins--Monro step size for policy updates \\
$p_H(a\mid s,T),\, p_{AI}(a\mid s,T)$ & Empirical (time-$T$) Human and AI policies \\
$\pi_H(\cdot\mid s),\, \pi_{AI}(\cdot\mid s)$ & Limiting stationary policies, defined as $\pi_i(\cdot\mid s)=\lim_{T\to\infty} p_i(\cdot\mid s,T)$ \\
$D_{JS}\!\big(p_H(\cdot\mid s,T)\,\|\,p_{AI}(\cdot\mid s,T)\big)$ & Jensen--Shannon divergence between Human and AI policies \\
$f_{\mathrm{agree}}$ & Long-run frequency of the agreement regime \\
$f_{\mathrm{ctx}}$ & Long-run frequency of the contextual regime \\
$f_{\mathrm{disagree}}$ & Long-run frequency of the disagreement regime \\
$P(\cdot \mid s,a)$ & Transition kernel of Virtual Nature governing state evolution \\
$S_T$ & Public entropy-based signal used to evaluate policy divergence \\
$T$ & Discrete decision-time index \\
$q_s(a,b)$ & Empirical action--reward signature distribution at state $s$ \\
$\delta_a$ & Dirac distribution concentrated on action $a$ (used in policy updates) \\
\bottomrule
\end{tabular}
\label{tab:Notation}
\end{table}
\FloatBarrier


\section{Base-case pseudocode}
\label{pseudocode}

This pseudocode contains the base case for two-threshold entropy delegation with simulated Human arbitration. 

\vspace{0.1cm}

\begin{algorithm}[H]
\footnotesize
\caption{Scenario 1: Simulated human arbitration}
\label{alg:basecase-s1-singleH}
\DontPrintSemicolon
\SetAlgoLined
\SetKwInOut{Input}{Input}\SetKwInOut{Output}{Output}

\Input{State space $S=\{A,B\}$; action set $\mathcal{A}=\{a_0,a_1\}$; horizon $H$; thresholds $0\le \alpha_{\mathrm{agree}}<\alpha_{\mathrm{disagree}}\le 1$;
EWMA rate $\beta$; discount $\gamma\in[0,1)$; step schedule $\{\eta_T\}$ (constant $\eta$: Markov-chain interpretation; $\eta_{\mathrm{decay}}$: Robbins--Monro);Markov transition kernel $P(\cdot\mid s,a)$; simulation flag \texttt{random\_states}; exogenous state probability $p_A$ (used when \texttt{random\_states=True}).}

\Output{Time series $D_{JS}^T$ and $f_H(T)$; policy components $p_H(a_0\mid A,T),p_H(a_0\mid B,T),p_{AI}(a_0\mid A,T),p_{AI}(a_0\mid B,T)$; EWMA utility traces $\hat U_H,\hat U_{AI}$; realised coalition utility $\hat U_{\mathrm{coal}}(T)$; value estimates $\hat V(\cdot)$; realised coalition value $\hat V_{\mathrm{coal}}(T)$; terminal policies.}

\BlankLine
\textbf{Init:} initialize policies \(p_H(\cdot \mid s,0)\), \(p_{AI}(\cdot \mid s,0)\) for each \(s \in S\); initialize utility traces \(\hat U_H(s,a) \leftarrow 0\) and \(\hat U_{AI}(s,a) \leftarrow 0\) for all \((s,a)\); initialize value function \(V(s) \leftarrow 0\) for each \(s \in S\); set counter \(C_H \leftarrow 0\).\;

\end{algorithm}

\begin{algorithm}[H]
\footnotesize
\caption*{Scenario 1: Simulated human arbitration (continued)}
\DontPrintSemicolon
\SetAlgoLined

\For{$T=1,2,\dots,H$}{
    Observe current state $s_T$.\;

    \tcp{Proposals from current policies. $p_{AI}$ and $p_H$ change T by T-1 to make the calculations since they already depend on $T-1$ see Eq (\ref{empirical_policies}).}
    Sample $a_H(T)\sim p_H(\cdot\mid s_T,T-1)$;\;
    Sample $a_{AI,T}\sim p_{AI}(\cdot\mid s_T,T-1)$;\;

    \tcp{Policy divergence at the current state (Jensen--Shannon)}
    Compute $D_{JS}^T\leftarrow D_{JS}\!\big(p_H(\cdot\mid s_T,T-1),\,p_{AI}(\cdot\mid s_T,T-1)\big)$.\;

    \BlankLine
    \tcc{Delegation/execution via two-threshold rule + simulated contextual rule)}
    \eIf{$D_{JS}^T\le \alpha_{\mathrm{agree}}$}{
        $\lambda_T(s) \leftarrow 0$;\tcp*{AI executes}
        $a_T^\ast \leftarrow a_{AI}(T)$;\;
    }{
        \eIf{$D_{JS}^T\ge \alpha_{\mathrm{disagree}}$}{
            $\lambda_T(s) \leftarrow 1$;\tcp*{Human executes}
            $a_T^\ast \leftarrow a_H(T)$;\;
        }{
            \tcc{Contextual zone: simulated arbitration (Equation~\eqref{lambda_context})}
            Draw $\kappa_T\sim \mathrm{Uniform}(0,1)$;\;
            $\lambda_T(s)\leftarrow 1$ if $\kappa_T<D_{JS}^T$, else $\lambda_T(s)\leftarrow 0$;\;
            $a_T^\ast \leftarrow a_H(T)$ if $\lambda_T(s)=1$, else $a_T^\ast \leftarrow a_{AI}(T)$;\;
        }
    }

    $C_H \leftarrow C_H + \lambda_T(s)$;\;
    $f_H(T)\leftarrow C_H/T$;\;

    \BlankLine
    \tcp{EWMA utility update at the executed pair For each $i\in\{H,AI\}$ and each $(s,a)\in S\times\mathcal A$:}
     $\hat U_i^{T+1}(s,a)=(1-\beta)\hat U_i^T(s,a)+
\beta\,\mathbf 1\{i=o_T \land (s,a)=(s_T,a_T^\ast)\}$.

    \BlankLine
    \tcp{Regime-conditioned immediate reward uses realized $\lambda_T(s)$}
    Define for any $a\in\mathcal{A}$:\;
    $r_T(s_T,a)\leftarrow \mathbf{1}\{\lambda_T(s)=1\}\hat U_H(s_T,a)+\mathbf{1}\{\lambda_T(s)=0\}\hat U_{AI}(s_T,a)$;\;

    \tcp{Virtual Nature kernel used in the Bellman expectation}
    \eIf{\texttt{random\_states} = True}{
        $P_{\mathcal{V}}(A\mid s,a)\leftarrow p_A$ and $P_{\mathcal{V}}(B\mid s,a)\leftarrow 1-p_A$ for all $(s,a)$;\;
    }{
        $P_{\mathcal{V}}(\cdot\mid s,a)\leftarrow P(\cdot\mid s,a)$ for all $(s,a)$;\;
    }

    \tcp{One-step Bellman optimality backup at the current state ($\lambda_T(s)$ as fixed at time $T$)}
    $\hat V(s_T)\leftarrow \max\limits_{a\in\mathcal{A}}
    \left\{ r_T(s_T,a) + \gamma \sum\limits_{s_{T+1}^{\mathrm{possible}}\in S}
    P_{\mathcal{V}}(s_{T+1}^{\mathrm{possible}}\mid s_T,a)\,
    \hat V(s_{T+1}^{\mathrm{possible}}) \right\}$;\;

    \tcp{In terms of calculations and states defined above, we write:}
    $\hat V(s_T)=\max_{a \in A_{\text{exec}}(s_T)}\left\{r(s_T,a;\lambda_T(s_T))
     + \gamma\left[P(A \mid s_T,a)\hat V(A) + P(B \mid s_T,a)\hat V(B)\right]\right\}$

     \tcp{Realised coalition utility at time $T$}
    $\hat U_{\mathrm{coal}}(T)\leftarrow
    \mathbf{1}\{\lambda_T(s)=1\}\hat U_H(s_T,a_T^\ast)+
    \mathbf{1}\{\lambda_T(s)=0\}\hat U_{AI}(s_T,a_T^\ast)$;\;
    
    \tcp{Realised coalition value at time $T$}
    $\hat V_{\mathrm{coal}}(T)\leftarrow \hat V(s_T)$;\;
    \BlankLine
    \tcc{Policy update: constant-step (Markov) or Robbins--Monro via decay}
    Set $\eta\leftarrow \eta_T$;\tcp*{e.g., $\eta_T=\eta$ (constant) or $\eta_T=\eta_{\mathrm{decay}}(T)$}
    Update $p_{AI}(\cdot\mid s_T,T)\leftarrow (1-\eta)p_{AI}(\cdot\mid s_T,T-1)+\eta\,\delta_{a_T^\ast}$;\;
    \If{$D_{JS}^T\le \alpha_{\mathrm{agree}}$}{
        Update $p_H(\cdot\mid s_T,T)\leftarrow (1-\eta)p_H(\cdot\mid s_T,T-1)+\eta\,\delta_{a_T^\ast}$;\;
    }

    \BlankLine
    \tcp{Environment transition (Virtual Nature)}
    \eIf{\texttt{random\_states} = True}{
        Sample $s_{T+1}=A$ with probability $p_A$, else $s_{T+1}=B$.\;
    }{
        Sample $s_{T+1}\sim P(\cdot\mid s_T,a_T^\ast)$.\;
    }
}
\end{algorithm}


\clearpage

\bibliographystyle{unsrt}
\bibliography{ref}

\end{document}